\documentclass[12pt]{article}

\usepackage[normalem]{ulem}
\usepackage{amssymb,amsmath,bm}
\usepackage{mathrsfs}
\usepackage{stackengine}
\usepackage{constants,color}
\usepackage{enumerate}
\usepackage{appendix}
\usepackage{relsize}
\usepackage{tikz}
\def\qed{\hfill \mbox{\rule{0.5em}{0.5em}}}
\usepackage{subcaption}

\usepackage{graphicx}
\usepackage{amssymb}
\usepackage{bookmark}
\usepackage{indentfirst}
\usepackage{cleveref}

\newcommand{\be}{\begin{equation}}
\newcommand{\ee}{\end{equation}}
\newcommand{\bes}{\begin{equation*}}
\newcommand{\ees}{\end{equation*}}
\newcommand{\ba}{\begin{aligned}}
\newcommand{\ea}{\end{aligned}}
\newcommand{\bi}{\begin{itemize}}
\newcommand{\ei}{\end{itemize}}

\newcommand{\EE}{\mathbb{E}}
\newcommand{\PP}{\mathbb{P}}

\newcommand{\var}{\mbox{Var}}

\allowdisplaybreaks

\newcommand{\va}{\mbox{Var}}

\numberwithin{equation}{section}

\newtheorem{theorem}{Theorem}
\newtheorem{corollary}{Corollary}

\newtheorem{proposition}{Proposition}
\newtheorem{lemma}{Lemma}

\newtheorem{remark}{Remark}[section]

\title{Parameter Estimation from Single Patient, Single Time-Point Sequencing Data of Recurrent Tumors}
\author{Kevin Leder$^{1}$ \and\hspace*{-6pt} Ruping Sun$^{2}$ \and\hspace*{-6pt}  Zicheng Wang$^{3}$ \and\hspace*{-6pt}Xuanming Zhang$^{1}$}
\date{%
    \footnotesize $^1$Department of Industrial and Systems Engineering, University of Minnesota, Twin Cities, MN 55455, USA. \\[2pt]
     $^2$Department of Laboratory Medicine \& Pathology Masonic Cancer Center, University of Minnesota, Twin Cities, MN 55455, USA. \\[2pt]
    $^3$School of Data Science, The Chinese University of Hong Kong, Shenzhen (CUHK-Shenzhen), China
}

\begin{document}
\maketitle
\begin{abstract}
In this study, we develop consistent estimators for key parameters that govern the dynamics of tumor cell populations when subjected to pharmacological treatments. While these treatments often lead to an initial reduction in the abundance of drug-sensitive cells, a population of drug-resistant cells frequently emerges over time, resulting in cancer recurrence. Samples from recurrent tumors present as an invaluable data source that can offer crucial insights into the ability of cancer cells to adapt and withstand treatment interventions. To effectively utilize the data obtained from recurrent tumors, we derive several large number limit theorems, specifically focusing on the metrics that quantify the clonal diversity of cancer cell populations at the time of cancer recurrence. These theorems then serve as the foundation for constructing our estimators. A distinguishing feature of our approach is that our estimators only require a single time-point sequencing data from a single tumor, thereby enhancing the practicality of our approach and enabling the understanding of cancer recurrence at the individual level.
\end{abstract}

\section{Introduction}
A significant obstacle in the effective treatment of cancer is the ability of cancer cells to develop resistance to pharmacological treatments. Drug resistance in cancer often stems from genetic mutations (single nucleotide variants or copy number alternations), which can modify the structure of proteins targeted by anticancer drugs, activate alternate signaling pathways circumventing the drug's action, or enhance the cancer cells' internal repair mechanisms \cite{Housman2014}. Mutated cells, endowed with such drug-resistant traits, gain a selective advantage over their drug-sensitive counterparts under treatment pressures. This clonal selection fosters the proliferation of drug-resistant cancer cells bearing diverse mutations, resulting in significant heterogeneity within recurrent tumors \cite{ding2012clonal, Hunter2006}. While such intratumor heterogeneity within recurrent tumors presents many challenges for subsequent therapy, it also offers an unique opportunity to decode the dynamics of cancer recurrence by examining the timing of occurrence and the tumor's clonal structure at cancer recurrence \cite{CD2021}.


High-throughput sequencing has enabled the reconstruction of the subclonal architecture of tumor samples. Initially, genomic sequencing profiles allow the accurate identification of somatic mutations in cancer cells, such as single nucleotide variants (SNVs) \cite{Cibulskis2013} and copy number alterations (CNAs) \cite{Ha2014}.  Subsequently, by clustering the cancer cell fractions associated with these mutations, indicating the proportion of cancer cells carrying each mutation, one can deduce the subclones present in the corresponding tumor sample \cite{Dentro2017}. Subclonal analysis methods utilize SNVs, CNAs, or a combination of both (cf. \cite{Roth2014} and \cite{VanLoo2010}). The ability to detect smaller subclones is influenced by sequencing depth and sampling strategy. For instance, deep sequencing ($\ge{500X}$) unveils finer subclones compared to shallower sequencing ($\le{50X}$) \cite{Stead2013}. Additionally, employing multiple tumor samples (or even single cells) for sampling provides a higher resolution than relying on a single sample \cite{Miller2014}. 
The methods we propose would require knowledge of the cancer cell fractions at the time of cancer recurrence.


Recent research has delved deeply into the mathematical modeling and analysis of cancer recurrence. Foo and Leder \cite{JK2013} presented a model that encompasses two types of cancer cells, investigating the dynamics of a subcritical branching process that avoids extinction. They identified two pivotal stochastic moments in cancer recurrence: the `turn-around time' where the total cell population begins to grow again, and the `crossover time' when resistant cells become dominant. Building on this foundation, Foo and colleagues \cite{foo2013cancer} expanded the model to account for heterogeneous escape populations, taking into account the variable fitness of resistant groups sourced from a mutational fitness landscape. This led to insights into the regrowth kinetics and the influence of distribution shape on recurrent tumor composition. Another study \cite{jasmine2012rare} delved deeper, focusing on rare events that precipitate early tumor recurrence, particularly the likelihood that the resistant cell population exceeds its average at the crossover time. Further contributions to this discourse include a study \cite{JKJ2014} that evaluates both random and deterministic fitness of resistant cells, providing evidence for a weak limit on the crossover time in both scenarios. Li and colleagues \cite{li2023comparison} investigated resistance mechanisms, both amplification-driven and mutation-driven. They established a law of large numbers that accounts for the convergence of the stochastic recurrence time across both mechanisms, defining this time as when the cancer cell population exceeds its original size. Avanzini and Antal \cite{avanzini2019cancer} delved into the issue of cancer relapse due to potential metastases. They introduced the concept of `recurrence time' as the stochastic moment when a metastasis first becomes detectable and derived the probability distribution for this occurrence. Taking the discourse further, Leder and Wang \cite{hanagal2022large,CD2021} enriched the understanding of cancer recurrence. Their focus was on examining the large deviations property of recurrence time and presenting an accurate depiction of the expected clonal diversity at recurrence, with a particular emphasis on early recurrence scenarios. Notably, while these studies have offered significant insights into recurrence time and the associated clonal heterogeneity, there remains an uncharted territory. There is a lack of strong convergence results vital for establishing consistent estimators related to metrics that detail clonal diversity, such as the number of clones or the Simpson's Index. This paper sets out to address this gap.




In the past decades, a growing body of research has been committed to developing estimators for parameters crucial in understanding tumor dynamics. Liang and Sha~\cite{liang2004modeling} utilized a deterministic, biexponential non-linear mixed-effects model to capture the potential tumor recurrence under treatment. They used tumor volume data from a study on mouse xenograft tumors to estimate the decay rates of the tumor's response to treatment. 
Lee and colleagues \cite{lee2022inferring} utilized a two-type branching process model to investigate the scenario where one or more initial wild-type cancer cells acquire an additional driver mutation as the tumor evolves. In their approach, they developed a novel methodology that leverages bulk sequencing data obtained at two separate time points, without any intervening treatment, to derive estimates for critical parameters in chronic lymphocytic leukemia. These include the net growth rate of the cancer cells and the timing of the appearance of the additional driver mutation. 
Werner and colleagues \cite{werner2020measuring} investigated a tumor population model characterized by exponential growth, employing coalescence theory to estimate the expected distribution of mutational distances from multi-sample sequencing data. In this context, mutational distances represent the mutation count disparities between two ancestral cells of different samples. Following this, they implemented the Metropolis-Hastings algorithm in conjunction with sequencing data from multiple samples at a single time point, enabling the estimation of both the mutation rate and cell survival rate throughout the expansion of the tumor population. 
Additional research efforts are currently delving into the estimation of parameters by leveraging data on mutation frequencies.
Salichos and colleagues \cite{salichos2020estimating} examined the variant-allele frequency of generational hitchhikers neutral mutations that arise prior to driver mutations. They used this data to infer the presence of subclonal drivers and estimate the fitness advantage conferred by these drivers.
In a series of works \cite{gunnarsson2021exact,gunnarsson2023limit}, Gunnarsson and colleagues adopted the Galton-Watson branching process to precisely describe the law of large numbers for the site frequency spectrum of neutral mutations, considering both large time and size limits. Utilizing these results, they formulated estimators for the extinction probability and mutation rate in a birth-death process. 
Our study differs from existing literature in that, unlike most previous research which primarily focuses on the net growth rate of tumor cells, our approach allows for the estimation of both the birth and death rates of tumor cells. Additionally, our study is distinct in providing theoretical guarantees regarding the consistency and convergence rates of the estimators, a crucial aspect often neglected in the current body of literature.

Lastly, this work is related to the stream of literature that studies the clonal structure of tumors (see for examples \cite{Cheek2020genetic}, \cite{Durrett2011intratumor}, \cite{Nicholson2023sequential}, \cite{Brouard2023genetic} and references therein). This stream of literature primarily seeks to characterize the sizes of all the mutant sub-populations in the large population limit. A focal point of these studies pertains to the mutation rates, wherein certain studies (e.g., \cite{Cheek2020genetic}) consider the rare mutation limit such that the product of the mutation rate $\mu_n$ and the targeted population size $n$ converges to a finite limit (i.e., $\lim\limits_{n\rightarrow \infty}n\mu_n=c<\infty$), while other studies (e.g., \cite{Brouard2023genetic}) consider the power law mutation rates limit such that $\mu_n\propto n^{-\alpha}$ with $\alpha\in (0,1)$. Our study adopts the latter perspective that mutation rates at the individual cell level are negligible ($\mu_n \ll 1$), but at the population level, the cumulative mutation rate is substantial ($n\mu_n \gg 1$), resulting in a high probability of the emergence of drug-resistant cells.


In this work, we adopt a two-type branching process model. The process starts with $n$ drug-sensitive tumor cells under pharmacological treatments, evolving over time as the tumor develops resistance through mutation with associated fitness changes. We define the recurrence time as the first time that the population of resistant cells reaches the initial tumor burden, $n$. We provide a set of convergence in probability results for various quantities observed at cancer recurrence. These include the recurrence time itself, denoted by $\gamma_n$, the number of mutant clones present at recurrence, denoted by $I_n(\gamma_n)$, and the Simpson's Index of mutant clones at recurrence, denoted by $R_n(\gamma_n)$.
The challenge in establishing these results arises from the intricate nature of the underlying stochastic process and the randomness of the recurrence time. Leveraging the convergence in probability results of these quantities, we construct estimators for the net growth rate of both sensitive and resistant cells, the birth rate of resistant cells, and the mutation rate. Compared to our previous work \cite{CD2021} and other related studies in the field, a distinguishing feature of our approach is that our estimators only require a single time-point sequencing data from a single tumor. Additionally, our estimators exhibit statistical consistency, thereby enhancing their reliability in characterizing individual tumors, which aids in personalized medical recommendations.

The organization of this paper is as follows. In Section 2, we introduce our model and previous results. In Section 3, we present the theoretical results, including the convergence in probability results for various quantities at cancer recurrence, and the consistency of our proposed estimators. In Section 4, we evaluate the performance of the estimators through simulation experiments. Lastly, we discuss the limitations of our approach and possible directions for further research in Section 5.

\section{Model and Previous Results}

\subsection{Basic Model underlying the Recurrence Process}\label{Sec:basic_model}


We assume there are two types of cancer cells: sensitive cells and resistant cells. At the time $t = 0$, the tumor is comprised solely of sensitive cells with an initial population size denoted by $n$, which we assume to be a known parameter. Typically, $n$ is an extremely large integer (a tumor with a volume of $1\text{ cm}^3$ is commonly assumed to contain $10^9$ cells). A treatment is applied at the time $t = 0$, which causes the number of sensitive cells to decrease over time. 

\begin{remark}
    It is important to note that the emergence of drug-resistant cells can occur before the initiation of pharmacological treatment, which suggests an alternative model that assumes a non-zero initial population of resistant cells \cite{hanagal2022large}. In a following work we investigate this alternative model.
\end{remark}


Here, we denote by $Z_0^n(t)$ the number of sensitive cells at time $t$ and assume $\left( Z_0^n(t) \right)_{t\geq 0}$ is a sub-critical birth-death process with birth rate $r_0$, death rate $d_0$, and net growth rate $\lambda_0 = r_0 - d_0 < 0$. We assume that at time $t$, sensitive cells also give birth to a resistant cell and a sensitive cell at rate $Z_0^n\left(t\right) n^{-\alpha}$ for $\alpha \in \left(0,1\right)$ (i.e., a mutation occurs). Each of these birth events results in the creation of a distinct clone of resistant cells which is modeled as a super-critical birth-death process with birth rate $r_1$, death rate $d_1$, and net growth rate $\lambda_1 = r_1 - d_1 > 0$. We denote the population of all resistant clones at time $t$ by $\left(Z_1^n\left(t\right)\right)_{t\ge 0}$ with $Z_1^n(0) = 0$.


The ability of resistant cells to proliferate under the treatment allows cancer cells to escape the extinction, resulting in the cancer recurrence. We define the cancer recurrence time as 
\begin{align}
    \label{eq:RecTimeDef}
    \gamma_n=\inf\{t\geq 0: Z_1^n(t)\geq \beta n\}, \text{ with }\beta>1.
\end{align}
In this paper, we adopt a simplified approach by setting $\beta =1$, thereby defining the cancer recurrence time as the first time at which the population of resistant cells reaches the initial population size of sensitive cells. We note that, in practice, the diagnosis of cancer recurrence might not coincide with the precise moment when the population of resistant cells reaches the initial tumor burden. However, our simplification does not substantially affect our subsequent analysis (see \cite{hanagal2022large} for more details).
\begin{remark}
An alternative definition of $\gamma_n$ is 
$$
 \gamma_n=\inf\{t> 0: Z_0^n(t)+Z_1^n(t)\geq n\},
$$ 
which represents the first time the total number of cancer cells exceeds the initial tumor size. However, this definition could lead to a very small value for $\gamma_n$ because of the randomness of the underlying process and the fact that the initial population size is already $n$. Nevertheless, if we ignore the initial fluctuations in population size and focus on the cancer recurrence driven by the growth of resistant cells, the difference between the two definitions becomes negligible, as resistant cells will dominate the population with high probability at the time of cancer recurrence.
\end{remark}
The model consists of $5$ unknown parameters which we denote by $\theta = \left(\alpha, r_0, d_0, r_1, d_1 \right)$. The objective of this paper is to develop estimators for the aforementioned parameters, using observable data during the cancer recurrence process.



\subsection{Observable Data during the Cancer Recurrence Process}
In Section \ref{Sec:basic_model}, we have defined the cancer recurrence time $\gamma_n$, which we assume to be an observable quantity clinically. Subsequently, we will present the definitions of several additional important quantities.


\paragraph{Number of Mutant Clones}
Each cell that has acquired resistance through mutation in the event of cell birth of sensitive cells gives rise to a distinct clone. Over time, some clones may die out while others continue to proliferate. For ease of exposition, we use ``mutant clones'' to denote these clones of resistant cells and use $I_n(t)$ to denote the number of mutant clones generated in the time period $\left(0,t\right)$ that has an infinite lineage (i.e., clones that do not go extinct).
\begin{remark}
\label{ref:number of mutant clones}
    To put it precisely, in a clinical or laboratory context, what we can observe is the number of mutant clones generated in the time period $\left(0,t\right)$ that survived until time $t$, which we denote by $\hat{I}_n(t)$. We opt for the current definition because it is easier to analyze and the difference between the two quantities ($I_n(t)$ and $\hat{I}_n(t)$) is negligible. In Section \ref{Sec:Convergence_RV}, we will make this statement precise and subsequently provide its proof.
\end{remark} 

\paragraph{Simpson's Index of Mutant Clones}
Simpson's Index measures the diversity in size among mutant clones. Specifically, it represents the probability that two resistant cells chosen at random come from the same clone. If we use $X_{i,n}(t)$ to denote the number of resistant cells in the $i_{th}$ clone at time $t$, then the Simpson's Index of mutant clones at time $t$, denoted by $R_n(t)$, is defined as 
\begin{align}
    \label{eq:SimIndDef}
    R_n(t)=\sum_{i=1}^{\hat{I}_n\left(t\right)}\left(\frac{X_{i,n}(t)}{Z_1^n\left(t\right)}\right)^2.
\end{align} 
For completeness, when $Z_1^n(t) = 0$, we define $R_n(t) = 0$.
\begin{remark}
Some of these metrics, such as Simpson's Index, might be observed imprecisely due to potential low-quality data using current experimental techniques and technology. Nonetheless, we expect that with technological advancements, it will be possible to measure the quantities more accurately.
\end{remark}

\begin{remark}
Define $\omega_n = \left\{ \gamma_n = \infty \right\}$, i.e., the event that the tumor never recurs. Note that when $\omega_n$ occurs, the observable quantities discussed in this section are ill-defined. However, it does not affect our results as 
\begin{align*}
    \lim_{n\to\infty}\PP\left(\gamma_n = \infty\right) = 0.
 \end{align*}
To show this, we first observe that
 \begin{align*}
     \PP\left(\gamma_n = \infty\right)& = \PP\left(Z_1^n(t) < n \text{\;for any \;} t>0 \right),
 \end{align*}
and the desired result follows directly from Corollary \ref{cor:Z1_larger_n}, which states that 
 \begin{align*}
      \lim_{n\to\infty}\PP\left(Z_1^n(k\zeta_n)<n\right) = 0
 \end{align*}
for $k>1$.     
\end{remark}
\subsection{Previous Results}
Let $\Phi^n_i(t) = \EE\left[ Z_i^n(t) \right]$, for $i = 0,1$. We have (cf. \cite{bp})
\begin{align}
    \Phi_0^n(t) &= ne^{\lambda_0t}, \nonumber\\
    \Phi_1^n(t) &=\frac{1}{\lambda_1-\lambda_0}n^{1-\alpha}e^{\lambda_1 t}\left(1-e^{\left(\lambda_0-\lambda_1\right)t}\right). \label{eq:expec_Z1}
\end{align}
It is not hard to find that $\Phi_1^n(t)$ is strictly increasing with $t$. Therefore there is a unique solution $\zeta_n$ for the equation
\begin{align*}
    \Phi_1^n(t) = n,
\end{align*}
and we can obtain that $\zeta_n \sim \frac{\alpha}{\lambda_1}\log{n}$.
From our previous work \cite{JK2013} \cite{CD2021}, we have
\begin{align}
& \lim_{n\to\infty}\PP\left(\left|\gamma_n - \frac{\alpha}{\lambda_1}\log{n}\right|>\epsilon\right) = 0, \label{cip_gamma}\\
& \lim\limits_{n\rightarrow \infty}\frac{1}{n^{1-\alpha}}\EE\left[I_n\left(\zeta_n\right)\right]=-\frac{\lambda_1}{\lambda_0 r_1}, \text{ and} \label{ciE_numofclone}\\
& \lim_{n\to\infty}n^{1-\alpha}\EE\left[R_n\left(\zeta_n\right)\right]= \frac{2r_1\left(\lambda_1-\lambda_0\right)^2}{ \lambda_1 \left(2\lambda_1-\lambda_0\right)}. \label{ciE_Simp}
\end{align}
In addition, we have
\begin{align}
\label{eq: exact_In}
    \EE\left[I_n\left(\zeta_n\right)\right] = -\frac{\lambda_1}{\lambda_0r_1}n^{1-\alpha}\left(1  - e^{\lambda_0\zeta_n}\right).
\end{align}


In equation \eqref{cip_gamma}, $\gamma_n$ is an experimentally observable quantity, whereas $\alpha$ and $\lambda_1$ are parameters in our model that we aim to estimate. This establishes a link between the model's `unobservable' parameters and an `observable' quantity. Our goal is to further explore these connections to develop consistent estimators for our model parameters. Equations \eqref{ciE_numofclone} and \eqref{ciE_Simp} are two other results that connect observable quantities (the number of mutant clones and Simpson's Index) with model parameters. However, these results have limited utility for constructing estimators in our context for two main reasons: Firstly, $\zeta_n$ is a deterministic approximation for $\gamma_n$ and it is not observable in practice. Secondly, estimators constructed based on the convergence of expectation require a large dataset (see \cite{CD2021}), which is not only costly but also does not ensure the consistency of the estimators. Furthermore, expectation based estimators do not allow for patient specific parameters and would instead estimate a single parameter set for the population. Consequently, we will extend these results (\eqref{ciE_numofclone} and \eqref{ciE_Simp}) to the random time $\gamma_n$ and establish the corresponding convergence in probability results.

\section{Theoretical Results}

In this section, we first present our main results on the convergence of `observable' data during the cancer recurrence processes (Section \ref{Sec:Convergence_RV}). Subsequently, leveraging these findings, we develop consistent estimators for the parameters of our model (Section \ref{Sec:Estimator}).

\subsection{Convergence Results of Observable Data}\label{Sec:Convergence_RV}
\label{results_of_1st_model}
We first establish a stronger convergence result for the cancer recurrence time $\gamma_n$. In \cite{JK2013}, the authors showed that $\gamma_n-\zeta_n$ converges to 0 in probability. In Theorem \ref{cip_gamma_faster}, we show that this result can be strengthened by specifying the convergence rate.
\begin{theorem}[Convergence rate of $\gamma_n$]\label{cip_gamma_faster}
    For any $\epsilon> 0, u<(1-\alpha)/2$, we have
    \begin{align}
        \lim_{n\to\infty}\PP\left(n^u\left|\gamma_n - \zeta_n\right|>\epsilon\right) = 0.
    \end{align}
\end{theorem}
\begin{proof}
    See Section \ref{sec:proof of cip_gamma_faster}.
\end{proof}

Next, we show a convergence in probability result for the number of mutant clones, $I_n(\gamma_n)$, observed at the time of cancer recurrence.
\begin{theorem}[Convergence of $I_n(\gamma_n)$]
    \label{cip_In}
    For any $\epsilon > 0, u<\min\{(1-\alpha)/2,  -\lambda_0\alpha/2\lambda_1 \}$,  we have
\begin{align}
    \label{eq:cip_In}
    \lim_{n\to\infty}\PP\left(n^u\left| \frac{1}{n^{1-\alpha}}I_n\left(\gamma_n\right)+\frac{\lambda_1}{\lambda_0r_1}\right| > \epsilon \right) = 0.
\end{align}
\end{theorem}
\begin{proof}
    See Section \ref{sec:proof of cip_In}.
\end{proof}


In the aforementioned Remark \ref{ref:number of mutant clones}, the term $I_n(t)$ represents the number of clones generated before time $t$ that do not go extinct. It is important to note that at any given time $t$, it is impossible to know whether a clone will go extinct or not. However, it can be demonstrated that the actual observation $\hat{I}_n(t)$, defined as the number of clones that emerge prior to time $t$ and remain extant at time $t$, has a negligible difference with $I_n(t)$.
\begin{proposition}
\label{prop:diff of In and hat_In}
    There exists $c_1>0$, such that for any $t>0$ 
    \begin{align*}
        \frac{1}{n^{1-\alpha}}\EE\left[ \hat{I}_n(t)- I_n(t) \right]\leq \frac{c_1d_1}{r_1(\lambda_0+\lambda_1)}\left( e^{\lambda_0 t}-e^{-\lambda_1 t} \right)
    \end{align*}
    if $\lambda_0+\lambda_1\neq 0$; and 
    \begin{align*}
        \frac{1}{n^{1-\alpha}}\EE\left[ \hat{I}_n(t)- I_n(t) \right]\leq \frac{c_1d_1}{r_1}te^{-\lambda_1 t}
    \end{align*} 
    otherwise.
\end{proposition}
\begin{proof}
We use $A(t)$ to denote the arrival process of mutant clones that are generated within the interval $(0,t)$ and go extinct eventually, but survive until time $t$. $A(t)$ is clearly a conditional non-homogeneous Poisson process with mean
\begin{align*}
    \EE\left[ A(t) \right] &= \EE\left[ \EE\left[ A(t)|Z_0(s), s\leq t \right] \right]\\
    & = \EE\left[ \int_0^tZ_0(s)n^{-\alpha}\PP(t-s<\tau_0<\infty)ds \right]\\
    & = \PP(\tau_0<\infty)\EE\left[ \int_0^tZ_0(s)n^{-\alpha}\PP(t-s<\tau_0|\tau_0<\infty)ds \right]\\
    & \overset{\text{(a)}}{=} \PP(\tau_0<\infty)n^{1-\alpha} \int_0^t e^{\lambda_0s} \PP(\tau_0>t-s|\tau_0<\infty)ds,
\end{align*}
where $\tau_0$ represents the extinction time of a mutant clone starting from a single resistant cell, and we use Fubini's theorem in equality (a). It is established in the literature (see the proof of Lemma 2 in \cite{durrett2016spatial}) when conditioning on the event of eventual extinction, i.e. $\tau_0< \infty$, a supercritical birth-death process becomes a subcritical process with birth and death rates interchanged. It implies that 
\begin{align*}
    \PP(\tau_0>t-s|\tau_0<\infty) = \PP(\tilde{\tau}_0 > t-s),
\end{align*}
where $\tilde{\tau}_0$ represents the extinction time for a subcritical birth-death process starting from a single cell with birth rate $d_1$ and death rate $r_1$. From Proposition 1 of \cite{jagers2007markovian}, for some $c_1>0$, the tail probability of the extinction time for the aforementioned subcritical birth-death process is bounded by
\begin{align*}
    \PP(\tilde{\tau}_0 > t)\leq c_1e^{-\lambda_1t}.
\end{align*}

Therefore, if $\lambda_0+\lambda_1\neq 0$, we have 
\begin{align*}
    \EE\left[ A(t) \right] & = \PP(\tau_0<\infty)n^{1-\alpha} \int_0^t e^{\lambda_0 s} \PP(\tau_0>t-s|\tau_0<\infty)ds \\
    & \leq \PP(\tau_0<\infty)n^{1-\alpha}e^{-\lambda_1 t} \int_0^t c_1e^{(\lambda_0+\lambda_1)s} ds \\
    & = \frac{c_1d_1}{r_1(\lambda_0+\lambda_1 )}n^{1-\alpha}\left( e^{\lambda_0 t}-e^{-\lambda_1 t} \right),
\end{align*}

which implies that
\begin{align*}
    \frac{1}{n^{1-\alpha}}\EE\left[ \hat{I}_n(t)- I_n(t) \right] = \frac{1}{n^{1-\alpha}}\EE\left[ A(t) \right]\leq \frac{c_1d_1}{r_1(\lambda_0+\lambda_1)}\left( e^{\lambda_0 t}-e^{-\lambda_1t} \right).
\end{align*}
The case where $\lambda_0+\lambda_1=0$ can be analyzed in a similar way and thus we omit the details here. \qed
\end{proof}

Utilizing Proposition~\ref{prop:diff of In and hat_In} and Theorem~\ref{cip_In}, we demonstrate that the actual observable number of clones also converges in probability.
\begin{corollary}
   \label{coro: In_hat}
    For any $\epsilon > 0, u<\min\{(1-\alpha)/2,  -\lambda_0\alpha/2\lambda_1\}$,  we have
\begin{align}
    \label{eq:cip_In}
    \lim_{n\to\infty}\PP\left(n^u\left| \frac{1}{n^{1-\alpha}}\hat{I}_n\left(\gamma_n\right)+\frac{\lambda_1}{\lambda_0 r_1}\right| > \epsilon \right) = 0.
\end{align}    
\end{corollary}
\begin{proof}
    By Theorem \ref{cip_In}, it suffices to show 
    \begin{align*}
    \lim_{n\to\infty}\PP\left(n^u\left( \frac{1}{n^{1-\alpha}}\hat{I}_n\left(\gamma_n\right)-\frac{1}{n^{1-\alpha}}I_n\left(\gamma_n\right)\right) > \epsilon \right) = 0.
\end{align*} 
Because $\hat{I}_n\left(\gamma_n\right) \ge I_n\left(\gamma_n\right)$, we can apply the Markov Inequality and obtain that
\begin{align*}
    &\PP\left(n^u\left( \frac{1}{n^{1-\alpha}}\hat{I}_n\left(\gamma_n\right)-\frac{1}{n^{1-\alpha}}I_n\left(\gamma_n\right)\right) > \epsilon \right)   \\
    \leq & \PP\left(n^u\left( \frac{1}{n^{1-\alpha}}\hat{I}_n\left(\gamma_n\right)-\frac{1}{n^{1-\alpha}}I_n\left(\gamma_n\right)\right) > \epsilon, \left|\gamma_n-\zeta_n\right|<\epsilon \right)+ \PP\left( \left|\gamma_n-\zeta_n\right|>\epsilon\right)\\
    = & n^{u-1+\alpha}\EE\left[ \hat{I}_n\left(\gamma_n\right)-I_n\left(\gamma_n\right); \left|\gamma_n-\zeta_n\right|<\epsilon \right]/\epsilon+ \PP\left( \left|\gamma_n-\zeta_n\right|>\epsilon  \right).
\end{align*}
 By Theorem \ref{cip_gamma_faster}, the second term converges to 0. Meanwhile, we have 
 \begin{align}
    &n^{u-1+\alpha}\EE\left[ \hat{I}_n\left(\gamma_n\right)-I_n\left(\gamma_n\right); \left|\gamma_n-\zeta_n\right|<\epsilon  \right] \nonumber\\
    \leq & n^{u-1+\alpha}\EE\left[ \hat{I}_n\left(\gamma_n\right)-I_n\left(\zeta_n-\epsilon\right); \left|\gamma_n-\zeta_n\right|<\epsilon  \right] \nonumber \\
    = & n^{u-1+\alpha}\EE\left[ \hat{I}_n\left(\zeta_n-\epsilon\right)-I_n\left(\zeta_n-\epsilon\right); \left|\gamma_n-\zeta_n\right|<\epsilon  \right] \label{eqn:hatI_I_diff}\\
    & \quad + n^{u-1+\alpha}\EE\left[ \hat{I}_n\left(\gamma_n\right)-\hat{I}_n\left(\zeta_n-\epsilon\right); \left|\gamma_n-\zeta_n\right|<\epsilon  \right] \label{eqn:hatI_hatI_diff}.
\end{align}
By Proposition \ref{prop:diff of In and hat_In}, when $u< \min\{(1-\alpha)/2,  -\lambda_0 \alpha/2\lambda_1\}$, we have  $\eqref{eqn:hatI_I_diff}\rightarrow 0$. By a similar argument to that in the proof of Lemma \ref{In_gamma_converge_to_In_zeta} (see Section \ref{sec:proof of cip_In}), we have $\ref{eqn:hatI_hatI_diff}\rightarrow 0$, which gives us
\begin{align*}
    \lim_{n\to\infty}\PP\left(n^u\left( \frac{1}{n^{1-\alpha}}\hat{I}_n\left(\gamma_n\right)-\frac{1}{n^{1-\alpha}}I_n\left(\gamma_n\right)\right) > \epsilon \right) = 0.
\end{align*}\qed 
\end{proof}

Next, we present a convergence in probability result for the Simpson's Index in Theorem \ref{cip_Rn}. We begin with the following proposition which shows that with high probability, resistant cells survive to the deterministic approximation of cancer recurrence time, $\zeta_n$, which paves the way for our later analysis for Simpson's Index. In particular define the set $\rho_n=\left\{ Z^n_1(\zeta_n)>0\right\}$ then we have the following result.
\begin{proposition}
\label{prop:rho_n}
 As $n\to\infty$, $\PP\left( \rho_n \right) \to 1.$
\end{proposition}
\begin{proof}
 From (\ref{Order: second moment of Z1}), we have
    \begin{align*}
        \PP\left( \rho_n^C \right) &= \PP\left(Z^n_1(\zeta_n)=0\right)\leq \PP(\left|Z_1^n(\zeta_n) -n\right|\geq n)\\
        &\leq \frac{\EE\left[\left(Z_1^n(\zeta_n)-n\right)^2\right]}{n^2}= \Theta(n^{\alpha-1})\to 0.
    \end{align*} \qed
\end{proof}

Furthermore, we have the following convergence in probability result for the Simpson's Index.
\begin{theorem}[Convergence of $R_n(\gamma_n)$]\label{cip_Rn}
    For any $\epsilon > 0, \alpha<1 , u<\min\{(1-\alpha)/4,(\lambda_1-\lambda_0 )\alpha/2\lambda_1\} $,  we have
\begin{align}
    \label{eq:cip_Rn}
\lim_{n\to\infty}\PP\left( n^u\left| n^{1-\alpha}R_n\left(\gamma_n\right)- \frac{2r_1\left(\lambda_1-\lambda_0\right)^2}{ \lambda_1 \left(2\lambda_1-\lambda_0\right)} \right| > \epsilon\right) = 0.
\end{align}
\end{theorem}
\textit{Proof Outline: }
From Proposition \ref{prop:rho_n} it will suffice to establish
$$
\lim_{n\to\infty}\PP\left( n^u\left| n^{1-\alpha}R_n\left(\gamma_n\right)- \frac{2r_1\left(\lambda_1-\lambda_0\right)^2}{ \lambda_1 \left(2\lambda_1-\lambda_0\right)} \right| > \epsilon, \rho_n\right)=0.
$$
To prove this result, we initially demonstrate the convergence of the Simpson Index at the deterministic time $\zeta_n$, as detailed in Lemma \ref{lemma:Convergence_of_Rn_zeta_n}. In Lemma \ref{lemma:Rn_gamma_to_Rn_zeta} we establish the proximity of $R_n(\gamma_n)$ to $R_n(\zeta_n)$. The detailed proofs can be found in Section \ref{sec: proof of convergence of Rn}.


Theorem \ref{cip_gamma_faster}, Theorem \ref{cip_In} and Theorem \ref{cip_Rn} ensure that, with high probability, the cancer recurrence time $\gamma_n$, the number of mutant clones $I_n(\gamma_n)$, and the Simpson's Index of mutant clones $R_n(\gamma_n)$ observed from a single patient sample are very close to their deterministic limits, respectively. In Section \ref{Sec:Estimator}, we use these results to develop a set of consistent estimators for our model parameters.


\subsection{Consistent Estimators for Model Parameters}\label{Sec:Estimator}

In this section, we develop a set of consistent estimators for our model parameters, including the growth rate of sensitive cells $\lambda_0$, the growth rate of resistant cells $\lambda_1$, the birth rate of resistant cells $r_1$, and the mutation rate $n^{-\alpha}$. These parameters are fundamental to the dynamics of cancer recurrence. The development of consistent estimators for these parameters can help practitioners gain a deeper understanding of the dynamics of cancer recurrence and for facilitating the decision-making process regarding patient treatment plans. For ease of notation, we use $Z_0^n \equiv Z_0^n(\gamma_n), I_n \equiv I_n(\gamma_n)$ and $R_n \equiv R_n(\gamma_n)$.

\begin{theorem}
\label{thm: consistent estimators}
For $\alpha < -\frac{\lambda_1}{\lambda_0 }$, define 
\begin{align*}
& \hat{\lambda}^{(n)}_0 = \frac{1}{\gamma_n}\log{\frac{Z^n_0}{n}},\\
& \hat{\lambda}^{(n)}_1 = -\frac{\hat{\lambda}^{(n)}_0}{U_n}, \\
& \hat{r}^{(n)}_1 = \left( \frac{1}{U_n}+1 \right)\frac{n\hat{\lambda}^{(n)}_1}{I_ne^{\hat{\lambda}^{(n)}_1\gamma_n}}, \text{ and}\\
&  \hat{\alpha}^{(n)} = 1 - \log_nI_n + \log_n\left( \frac{\hat{\lambda}^{(n)}_1}{- \hat{\lambda}^{(n)}_0\hat{r}^{(n)}_1} \right),
\end{align*}
where $U_n =  \frac{\sqrt{I_n\cdot R_{n}}}{\sqrt{I_n\cdot R_{n}-2}} - 1$. Then, $\hat{\lambda}^{(n)}_0, \hat{\lambda}^{(n)}_1, \hat{r}^{(n)}_1$ and $\hat{\alpha}^{(n)}$ are consistent estimators of $\lambda_0, \lambda_1, r_1$ and $\alpha$ respectively. 
\end{theorem}
\begin{proof}
     See Section \ref{sec: consistent estimators}.
\end{proof}

Theorem \ref{thm: consistent estimators} establishes the consistency of our estimators, ensuring their effectiveness when the initial size of the tumor is sufficiently large. Moreover, for $\hat{\lambda}^{(n)}_0$ and $\hat{\lambda}^{(n)}_1$, a lower bound can be established for the rate of convergence for both estimators. In particular, in the proof of Theorem \ref{thm: consistent estimators}, we obtain that 
\begin{align*}
    &\lim_{n\to \infty}\PP\left( n^{u}\left| \hat{\lambda}^{(n)}_0-\lambda_0  \right| > \epsilon\right) =0  \text{\quad for \;} u < \min\{(1-\alpha)/2, 1/2+\alpha\lambda_0/2\lambda_1\}, \text{and}\\
    &\lim_{n\to \infty}\PP\left( n^{u}\left| \hat{\lambda}^{(n)}_1-\lambda_1  \right| > \epsilon\right) =0  \text{\quad for \;} u < \min\{(1-\alpha)/4,  -\lambda_0\alpha/2\lambda_1,1/2+\alpha\lambda_0/2\lambda_1 \}.
\end{align*}
To the best of our knowledge, we are among the first in presenting an analysis of the convergence rate for estimators of key parameters in tumor dynamics. As will be discussed in Section \ref{Sec: simulation and extensions}, this analysis leads to an accurate estimation of these parameters within practical and realistic settings.

\begin{remark}
It should be noted that the formulation of our estimators utilizes data on the number of mutant clones and their Simpson's Index at the recurrence time. This data yields the same information as the mutant clone size's empirical mean and second moment. Specifically, while it is possible to devise estimators based on the first and second moments of mutant clone size, our preference is to employ the number of clones and Simpson's Index, which offers a more nuanced understanding of heterogeneity.
\end{remark}

\section{Simulation Results and Extensions}\label{Sec: simulation and extensions}

In this section, we elucidate the performance of our estimators through numerical simulation.

\subsection{Convergence of Estimators}

Theoretically, our estimators are posited to work with a sample of recurrent tumor from a patient, provided several conditions are met. These include the survival of sensitive cells until the time of cancer recurrence, a sufficiently large tumor size, and a constraint on the model parameters ($0< \alpha< \min\left\{-\lambda_1/\lambda_0, 1\right\}$). In this section, we present simulation results and a detailed discussion aimed at identifying the threshold for tumor size that ensures a good performance of our estimators.

We consider an example system with $\alpha = 0.5, r_0 = 0.5, d_0 = 1, r_1 = 1.5$ and $d_1 = 1$. Note that $\lambda_0 = r_0 - d_0 = -0.5$ and $\lambda_1 = r_1 - d_1 = 0.5$. Recall that we denote by $n$ the initial tumor burden (i.e., the initial population size of sensitive cells). We measure the performance of our estimators for $n=100\times 2^i$, where $i\in \{0,1,\cdots, 20\}$. Specifically, for each value of $n$, we simulate the corresponding system until the population of resistant cells reaches the initial tumor burden $n$. We then record the recurrence time ($\gamma_n$), the population size of sensitive cells at cancer recurrence ($Z_0(\gamma_n)$), the number of mutant clones ($I_n(\gamma_n)$), and the Simpson's Index ($R_n(\gamma_n)$). By Theorem \ref{thm: consistent estimators}, we compute the estimators $\hat{\lambda}_0^{(n)},\hat{\lambda}_1^{(n)},\hat{r}_1^{(n)}$ and $\hat{\alpha}_1^{(n)}$, and calculate the relative errors of these estimators, i.e.,
\begin{align*}
    & \lambda_0^{error} = \frac{\left| \hat{\lambda}_0^{(n)} - \lambda_0 \right|}{-\lambda_0},\; \lambda_1^{error} = \frac{\left| \hat{\lambda}_1^{(n)} - \lambda_1 \right|}{\lambda_1},\;  \alpha^{error} = \frac{\left| \hat{\alpha}^{(n)} - \alpha \right|}{\alpha},\; r_1^{error} = \frac{\left| \hat{r}_1^{(n)} - r_1 \right|}{r_1}.
\end{align*}
For each value of $n$, we repeat such a process for $k = 40$ times and report the mean value and standard deviation of $\lambda_0^{error}, \lambda_1^{error}, \alpha^{error}$ and $r_1^{error}$ in Figure \ref{fig:large_n}.
\begin{figure}[h!]
\centering
\setkeys{Gin}{width=1\linewidth}
\begin{minipage}[t]{0.5\textwidth}
\includegraphics{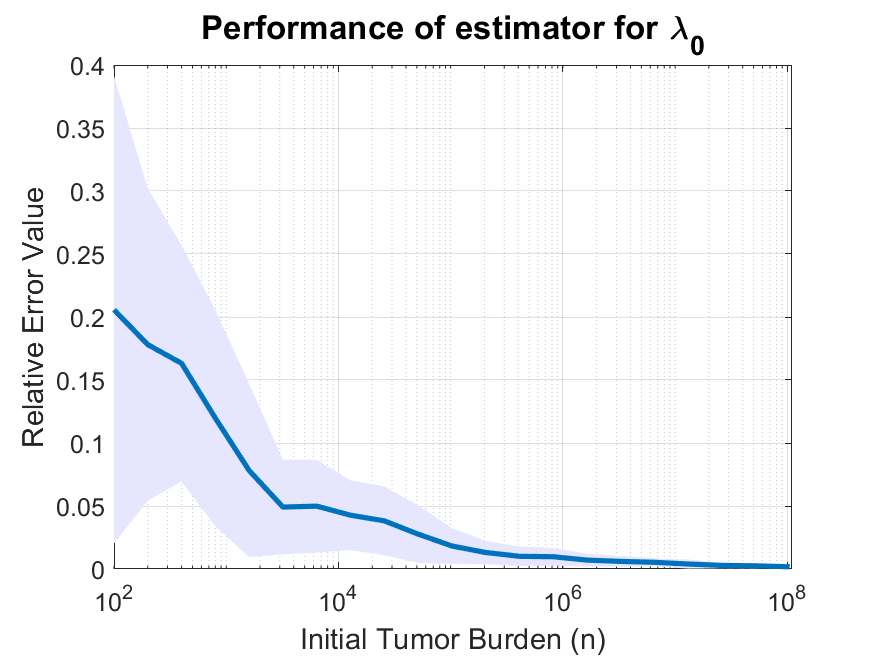}
\end{minipage}\hfill
\begin{minipage}[t]{0.5\textwidth}
\includegraphics{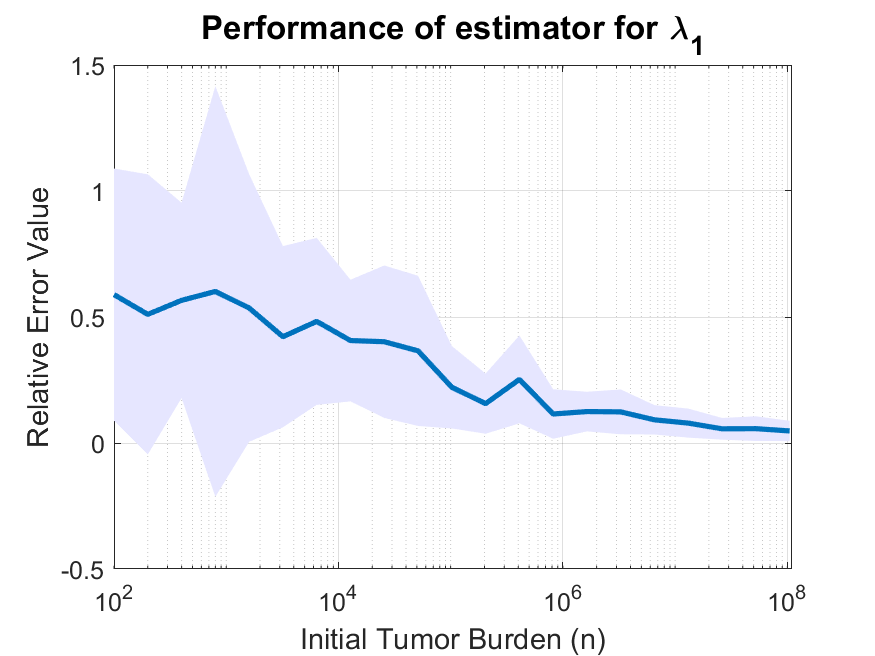}
\end{minipage}\hfill
\begin{minipage}[t]{0.5\textwidth}
\includegraphics{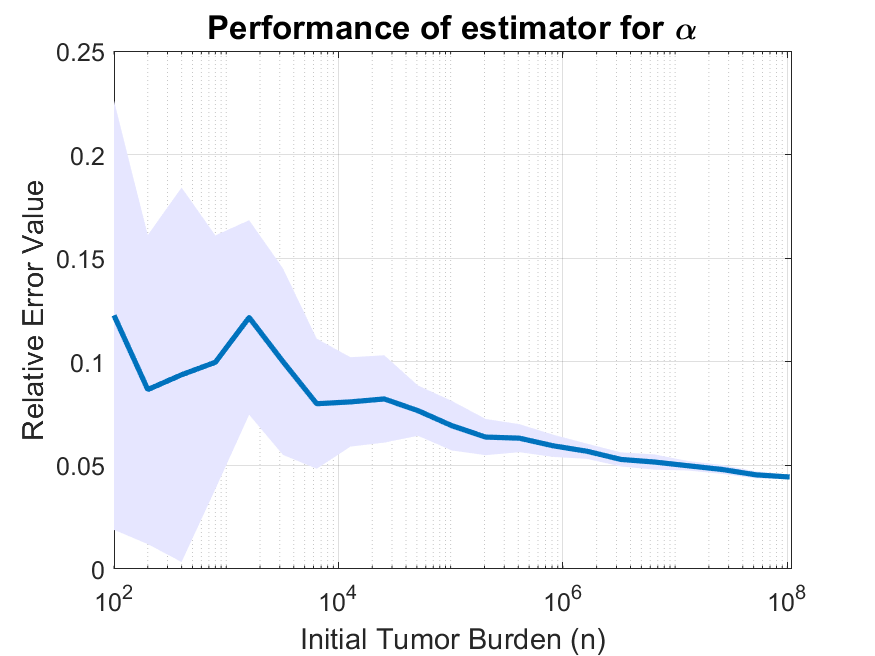}
\end{minipage}\hfill
\begin{minipage}[t]{0.5\textwidth}
\includegraphics{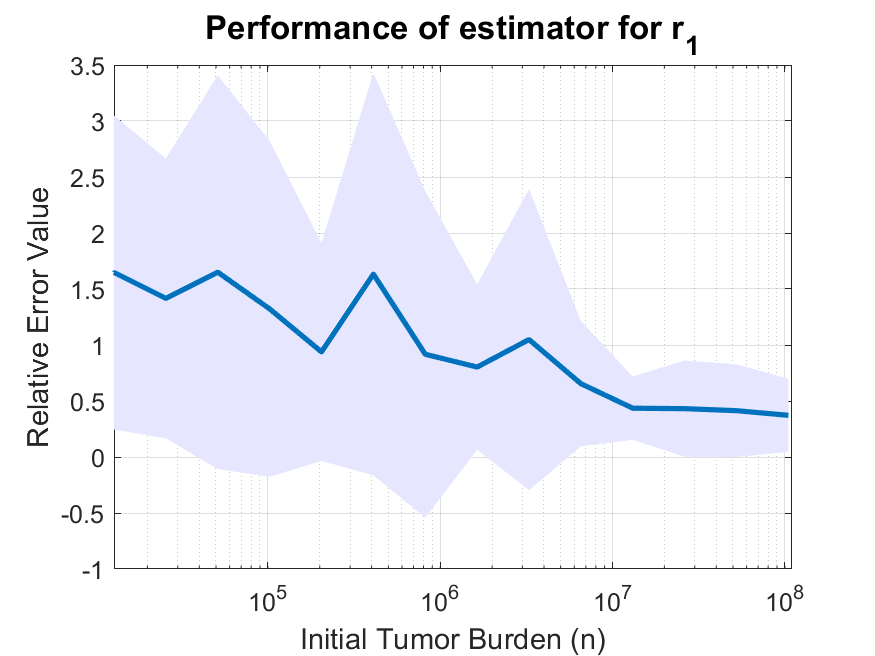}
\end{minipage}
\caption{In this experiment, we set $\alpha = 0.5, \lambda_0 = -0.5, r_1 = 1.5, d_1 = 1$ and $\lambda_1 = 0.5$. We let $n$ vary from $100$ to $100\times 2^{20}$. For each value of $n$, we repeat the experiment $k = 20$ times. We represent the mean relative errors of each estimator with a dark blue line and illustrate the standard deviation using a shaded area in light blue.}
\label{fig:large_n}
\end{figure}

Observations drawn from Figure \ref{fig:large_n} reveal that the mean values of relative errors associated with all estimators exhibit a decreasing trend towards $0$ as $n$ increases. Concurrently, the standard deviations, represented by the light blue shaded areas in the figure, similarly diminish to $0$ with the increase in $n$. These observations align with the theoretical findings presented in Theorem \ref{thm: consistent estimators}, confirming the consistency of our estimators.


Notably, these estimators exhibit dissimilar levels of performance. In particular, $\hat{\lambda}_0^{(n)}$ demonstrates superior performance, an observation that aligns with the theoretical results established in the proof of Theorem \ref{thm: consistent estimators}, which states that $\hat{\lambda}_0^{(n)}$ approaches $\lambda_0$ at a convergence rate of at least $n^{-\min\left\{  (1-\alpha)/2, (1+\lambda_0\alpha/\lambda_1)/2\right\}}$. However, the convergence rate of $\hat{\lambda}_1^{(n)}$ is only ensured to be at least $n^{-\min\{(1-\alpha)/4, -\lambda_0\alpha/2\lambda_1,(1+\lambda_0\alpha/\lambda_1)/2 \}}$. On the other hand, there are no assurances about the convergence rates of ${\alpha}^{(n)}$ and $\hat{r}_1^{(n)}$. Additionally, we note that the rate at which $\hat{r}_1^{(n)}$ converges is the slowest. The convergence of $\hat{r}_1^{(n)}$ depends on the convergence of the term $e^{(\lambda_1 - \hat{\lambda}_1^{(n)})\zeta_n}$ (see the proof of Theorem \ref{thm: consistent estimators} in Section \ref{sec: consistent estimators}). However, based on the convergence result of $\hat{\lambda}_1^{(n)}$, the theoretical guarantee for the convergence of $(\lambda_1 - \hat{\lambda}_1^{(n)})\zeta_n$ only holds when $\zeta_n$ is less than $n^{\min\{(1-\alpha)/4,  -\lambda_0\alpha/2\lambda_1 \}}$. Given that $\zeta_n$ is approximately $\frac{\alpha}{\lambda_1}\log n$, in order to ensure that the condition is satisfied, the value of $n$ must be at least $10^{11}$.  This magnitude surpasses the capabilities of our current simulation framework. However, it is within the realm of possibility for certain types of tumors to achieve such a numerical scale. A 10-cm cancerous tumor is estimated to contain approximately \(10^{11}\) cells~\cite{narod2012disappearing}. This size is not uncommon in various tumor types, with instances documented in breast carcinoma~\cite{giuliano1995improved} as well as nasopharyngeal carcinoma~\cite{sze2004primary}, among others.


\subsection{Performance of Estimators in the Presence of Undetectable Clones}

In this study, we have demonstrated, both theoretically and via simulation, the good performance of our estimators of the model parameters \( \lambda_0 \), \( \lambda_1 \), \( \alpha \), and \( r_1 \) in contexts where there is a large initial population size $n$ of drug sensitive cells. However, the accuracy of our estimators relies heavily on calculating the Simpson Index. This calculation requires detailed information about the sizes of all mutant clones at the time of recurrence. A significant challenge arises due to the limitations of current genomic sequencing technologies, which are incapable of identifying mutant clones of small size.
Previous studies \textcolor{blue}{\cite{Roth2014, Stead2013}} indicate that the minimum detectable clone size using current technologies is around \( 2\% \) of the total tumor size. Consequently, it is necessary to modify the Simpson Index computation in order to account for these undetectable clones, using only the data of detectable clones. To assess the effectiveness of this strategy, we conduct a simulation experiment that takes into account the presence of clones that fall below the detection threshold.

In our simulation experiments, we fix \( \alpha = 0.8 \), \( r_0 = 1.3 \), \( d_0 = 1.5 \), \( r_1 = 2.0 \), and \( d_1 = 1.2 \),  while progressively increasing the initial tumor burden \( n \). Each experiment consists of three distinct scenarios for evaluating our estimators: a baseline control and two test settings. In test settings, we exclude clone data that are smaller than \( 2\% \) and \( 10\% \) of the total tumor size at the time of recurrence, respectively. This approach allows us to assess the performance of our estimators in scenarios with varying degrees of data availability.
\begin{figure}[h!]
\centering
  \begin{subfigure}{8cm}
  \centering\includegraphics[width=1\textwidth]{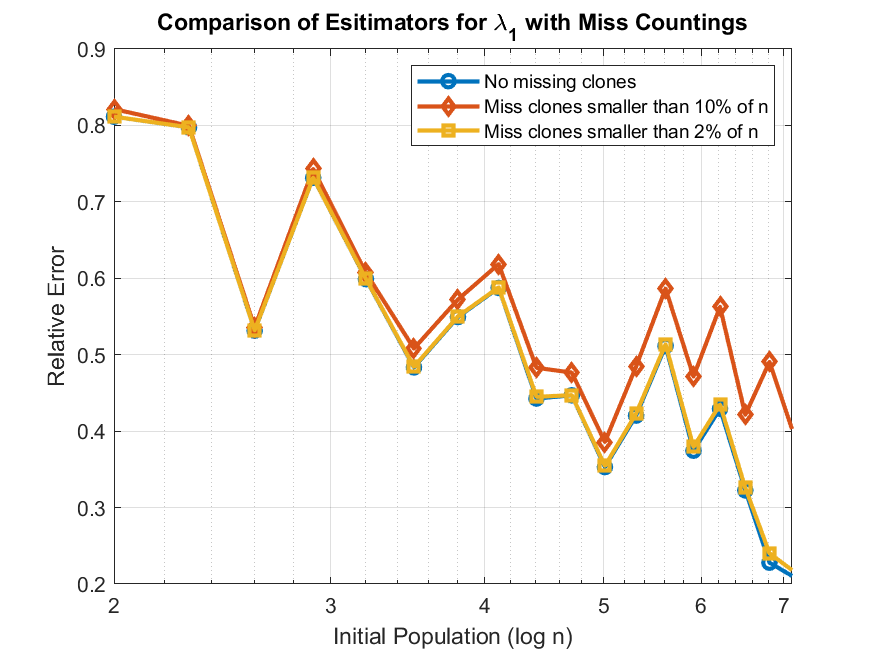}
  \end{subfigure}
  \begin{subfigure}{8cm}
    \centering\includegraphics[width=1\textwidth]{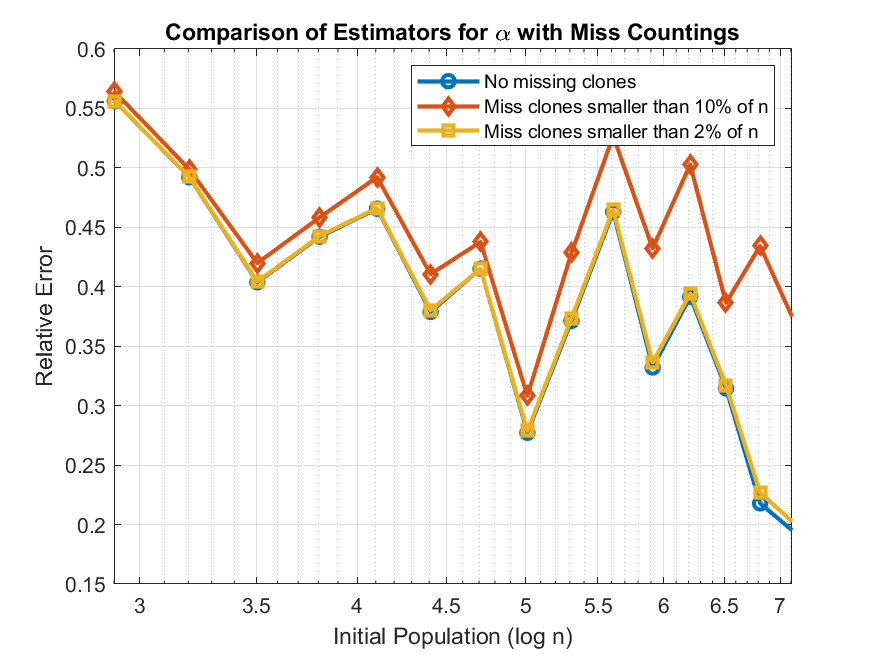}
  \end{subfigure}
    \begin{subfigure}{8cm}
  \centering\includegraphics[width=1\textwidth]{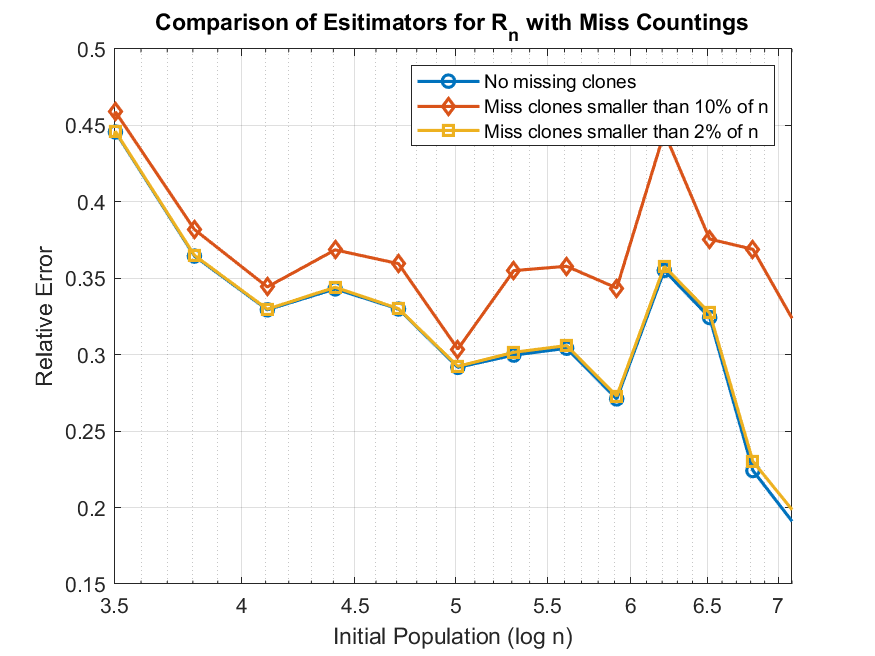}
  \end{subfigure}
  \begin{subfigure}{8cm}
    \centering\includegraphics[width=1\textwidth]{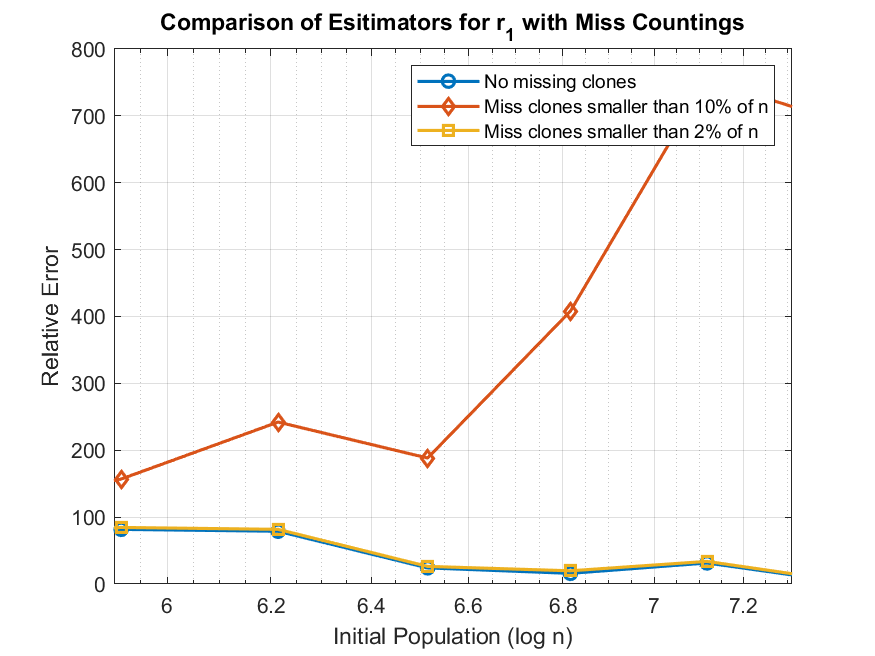}
  \end{subfigure}

\caption{In this experiment, we conduct a comparative analysis across three experimental setups. We set $\alpha = 0.8$, $\lambda_0 = -0.2$, $\lambda_1 = 0.8$ and $r_1 = 2.0$. The first setup is a baseline scheme without restrictions on clone size observation; the second setup imposes a threshold, only considering clones that constitute at least $2\%$ of the total population; and the third raises this threshold to $10\%$. We assess the impact of these varying observation restrictions on the performance of the estimators for $\lambda_1$ and $\alpha$.}
\label{fig:miss_count}
\end{figure}

As illustrated in Figure~\ref{fig:miss_count}, excluding clone size data below $10\%$ of the total tumor population may significantly affect the performance of our estimators. In contrast, the omission of clones smaller than $2\%$ of the total tumor population appears to have a negligible impact. This observation implies that, given the capabilities of current technologies to detect clones as small as $2\%$ of the total tumor size, the estimators constructed using our method are still reliable.

\subsection{Stability Analysis} 

As established in Theorem~\ref{thm: consistent estimators}, the performance of our estimators relies on the assumption that \( 0 < \alpha < \min\{1, -\lambda_1 / \lambda_0\} \). In this section, we examine the robustness of our estimators across a range of parameter settings, particularly near the critical boundary conditions, such as when $\alpha$ is close to either $\min\{1, -\lambda_1 / \lambda_0\}$ or $0$. To facilitate a clear interpretation of the results, we conduct three sets of numerical experiments. For each set, we modify a single parameter while keeping the others constant, enabling us to isolate the impact of each parameter.

In the first set of experiments, we fix \( r_1 = 1.5 \), \( d_1 = 1.0 \), and \( \alpha = 0.5 \). The values for \( r_0 \) and \( d_0 \) are uniformly sampled from $(0.8, 1.2)$ and $(1.3, 1.7)$ respectively, ensuring that \( \lambda_0 \) falls within $(-0.9, -0.1)$. In this case, we have \( -\lambda_0 < \lambda_1 / \alpha = 1 \). In the second set of experiments, we fix \( r_0 = 1.0 \), \( d_0 = 1.5 \), and \( \alpha = 0.5 \). The values for \( r_1 \) and \( d_1 \) are uniformly sampled from $(1.4, 1.8)$ and $(0.7, 1.1)$ respectively, ensuring that \( \lambda_1 \) falls within $(0.3, 0.9)$. In this case, we have \(\lambda_1 > -\alpha\lambda_0 = 0.25 \). In the third set of experiments, we fix \( r_0=1 \), \( d_0=1 \), \( r_1 =1.5\), and \( d_1=1.5 \) respectively. The value for \( \alpha \) is uniformly sampled from $(0, 1)$.

\begin{figure}[h!]
\setkeys{Gin}{width=1\linewidth}
\begin{minipage}[t]{0.25\textwidth}
\includegraphics{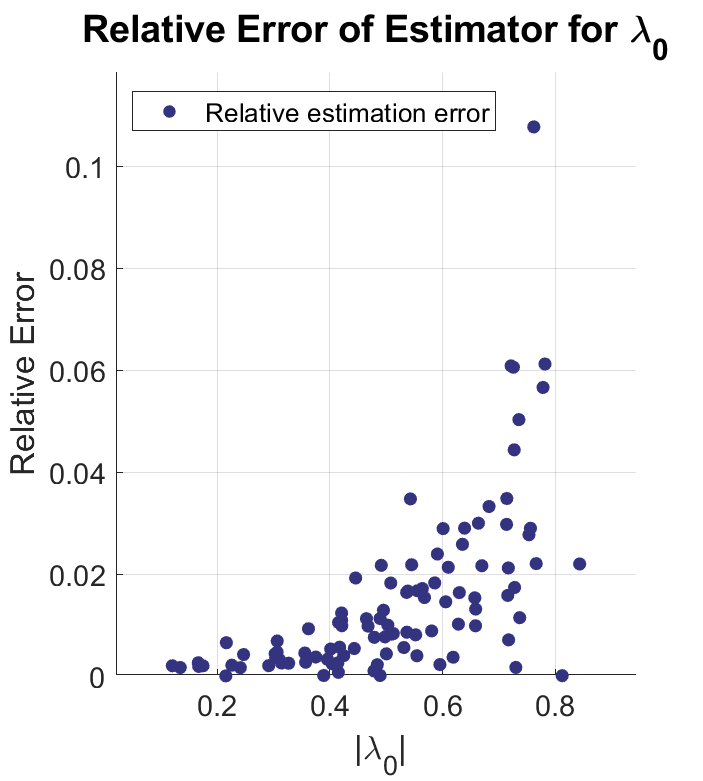}
\end{minipage}\hfill
\begin{minipage}[t]{0.25\textwidth}
\includegraphics{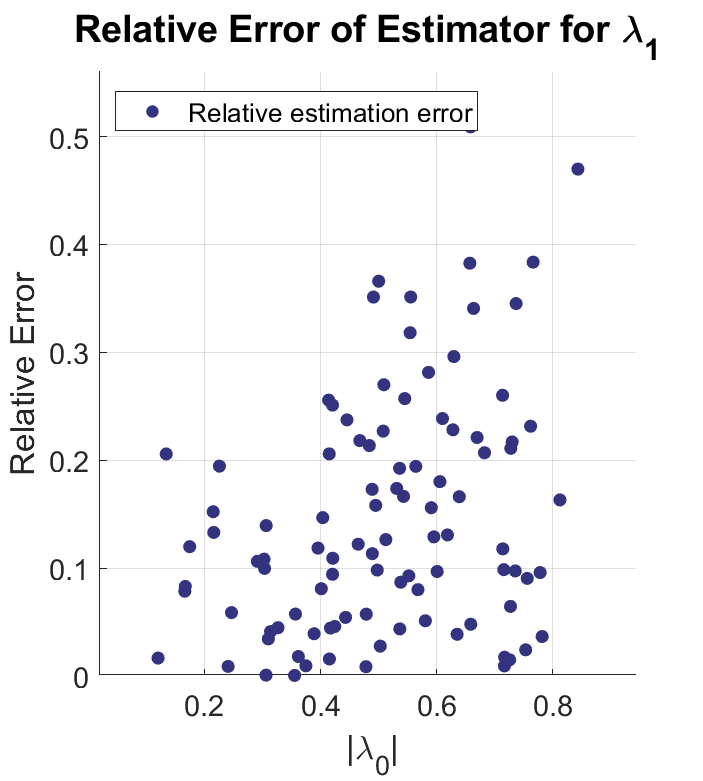}
\end{minipage}\hfill
\begin{minipage}[t]{0.25\textwidth}
\includegraphics{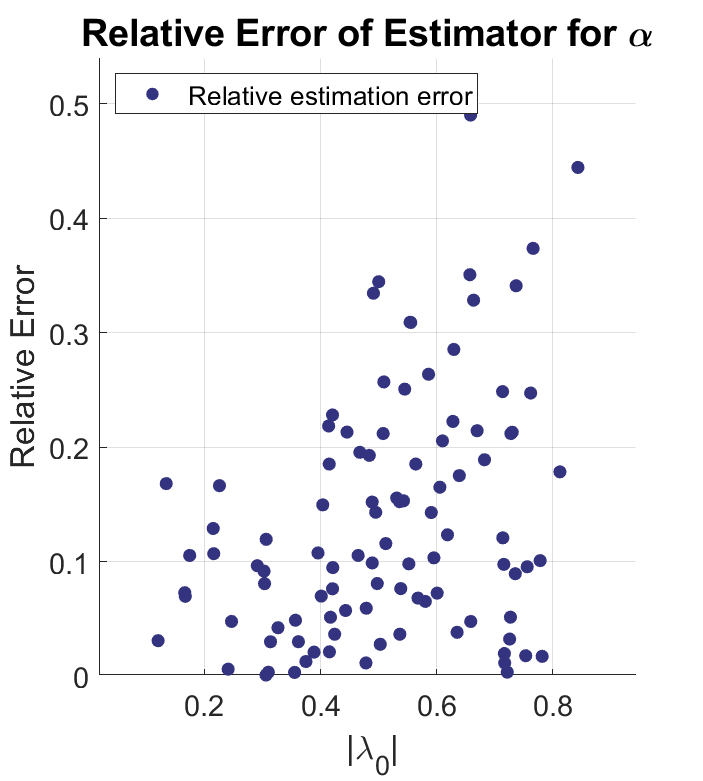}
\end{minipage}\hfill
\begin{minipage}[t]{0.25\textwidth}
\includegraphics{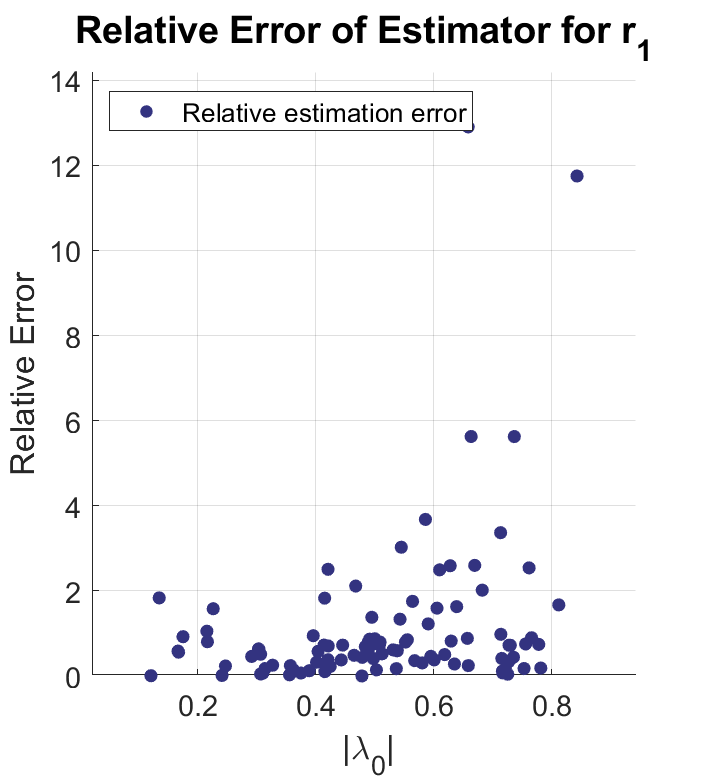}
\end{minipage}
\begin{minipage}[t]{0.25\textwidth}
\includegraphics{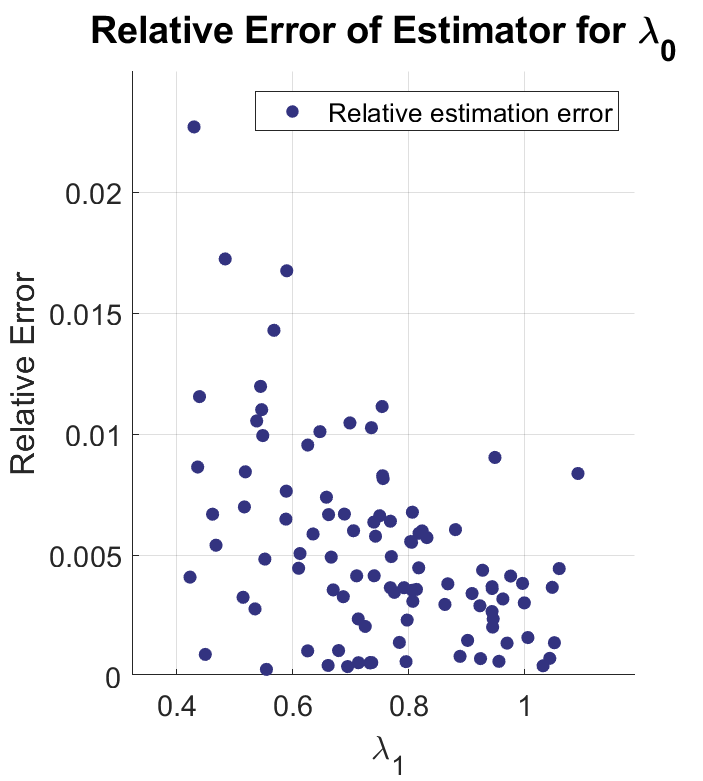}
\end{minipage}\hfill
\begin{minipage}[t]{0.25\textwidth}
\includegraphics{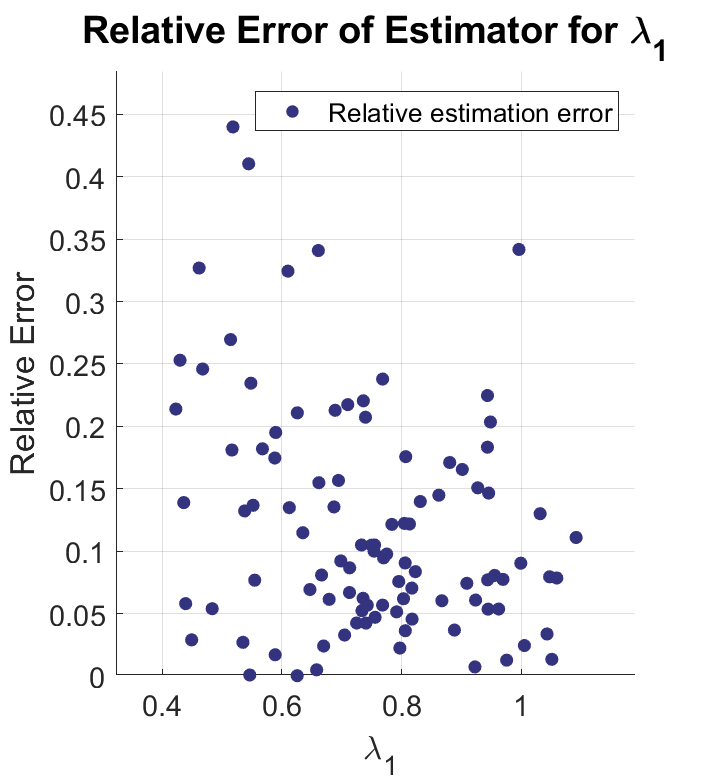}
\end{minipage}\hfill
\begin{minipage}[t]{0.25\textwidth}
\includegraphics{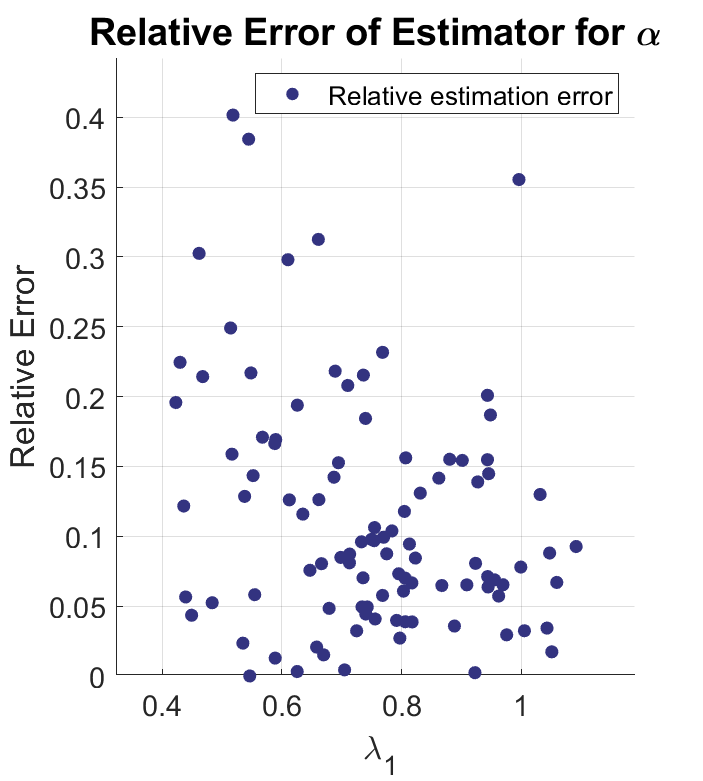}
\end{minipage}\hfill
\begin{minipage}[t]{0.25\textwidth}
\includegraphics{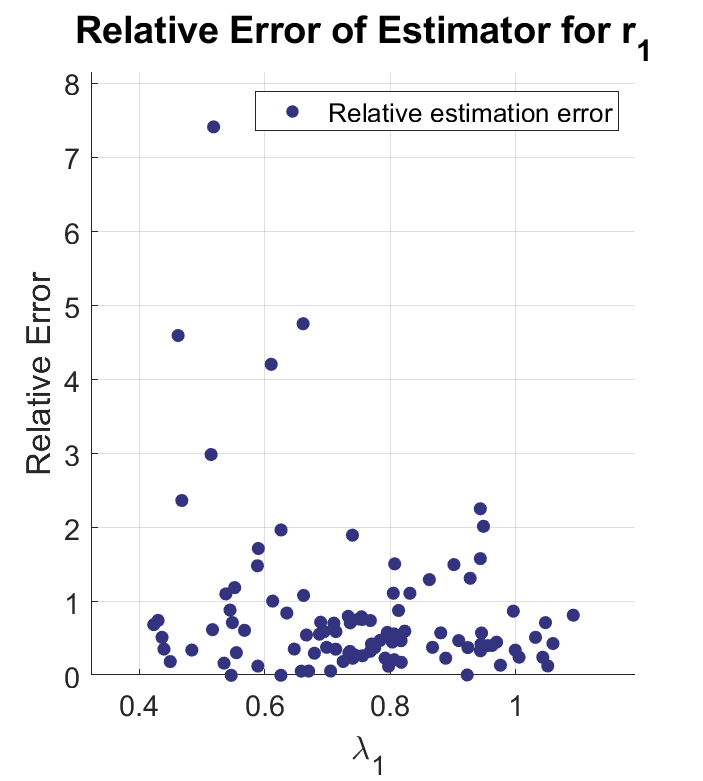}
\end{minipage}
\begin{minipage}[t]{0.25\textwidth}
\includegraphics{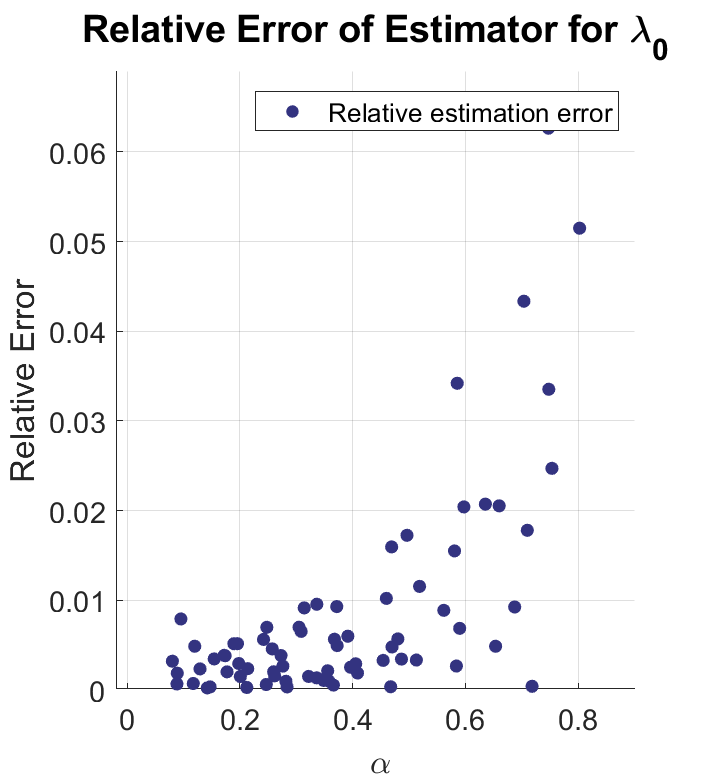}
\end{minipage}\hfill
\begin{minipage}[t]{0.25\textwidth}
\includegraphics{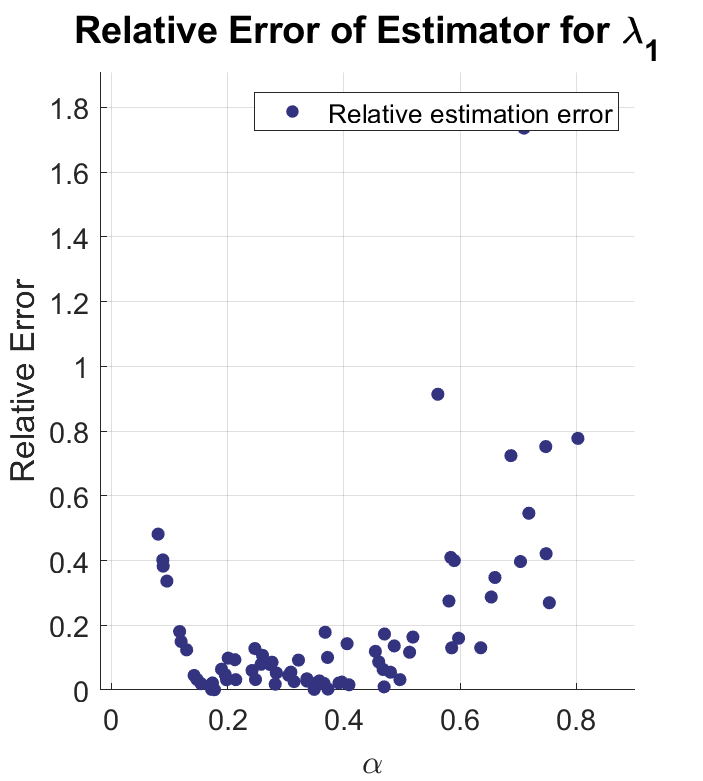}
\end{minipage}\hfill
\begin{minipage}[t]{0.25\textwidth}
\includegraphics{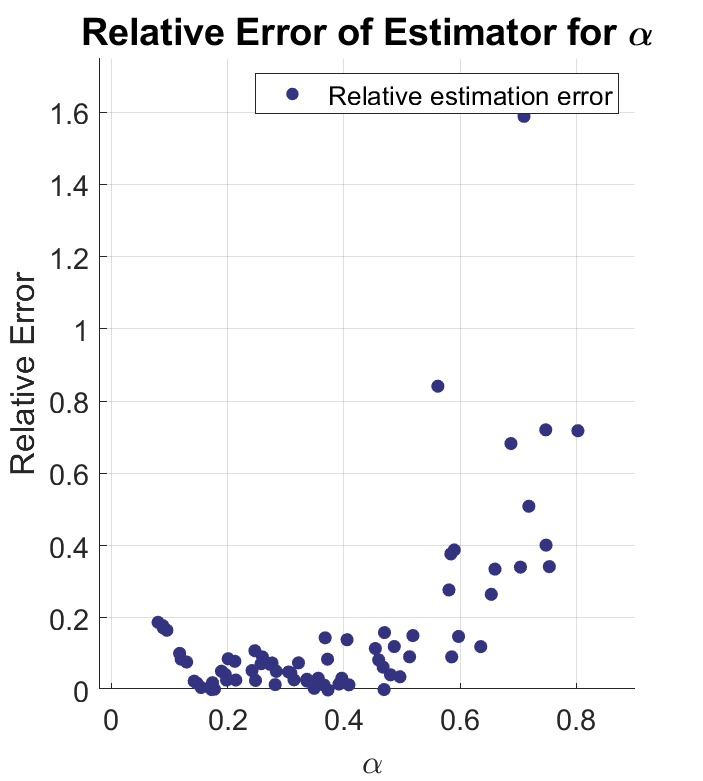}
\end{minipage}\hfill
\begin{minipage}[t]{0.25\textwidth}
\includegraphics{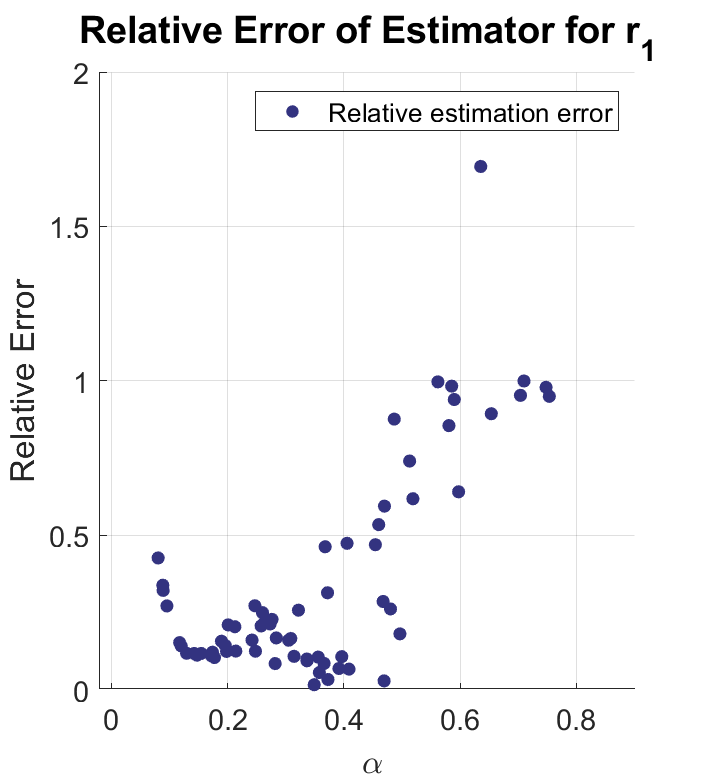}
\end{minipage}
\caption{In this experiment, we vary the model parameters and check the stability of our estimators. In the first setup, We fix $r_1 = 1.5$, $d_1 = 1.0$, and $\alpha = 0.5$, and select $\lambda_0$ uniformly from the interval $(-0.9, -0.1)$ (see the first row of Figure \ref{fig:random_para}). In the second setup, we fix $r_0 = 1.0$, $d_0 = 1.5$, and $\alpha = 0.5$ and select $\lambda_1$ uniformly from the interval $(0.3, 0.9)$ (see the second row of Figure \ref{fig:random_para}). Lastly, we set $r_0 = d_0 = 1.0$ and $r_1 = d_1 = 1.5$, and select $\alpha$ uniformly from the interval $(0, 1)$ (see the third row of Figure \ref{fig:random_para}). Under all setups, we fix the initial tumor burden to be \( n = 1 \times 10^6 \), and conduct numerical simulations to estimate \( \lambda_0 \), \( \lambda_1 \), \( \alpha \), and \( r_1 \). The resulting relative errors for these estimations, corresponding to each of the defined setups, are presented in Figure \ref{fig:random_para}.}
\label{fig:random_para}
\end{figure}


The numerical results depicted in Figure~\ref{fig:random_para} demonstrate that our estimators exhibit better performance in certain scenarios compared to others. These include scenarios where the absolute value of the growth rate of sensitive cells, \( |\lambda_0| \), is low; the growth rate of resistant cells, \( \lambda_1 \), is high; or the mutation rate, \( \alpha \), is moderately small. A mathematical interpretation of this result can be linked to the boundary condition.  By defining the parameter gap as \( \Delta = 1 - \alpha|\lambda_0|/\lambda_1 \), it becomes evident that a decrease in \( |\lambda_0| \) or an increase in \( \lambda_1 \), or a decrease in $\alpha$, results in a reduction of the parameter gap \( \Delta \). Note that
\begin{align*}
    Z_0^n(\gamma_n) \approx \Phi_0^n(\zeta_n) \sim n^{\Delta}.
\end{align*}
Hence, the population size of sensitive cells at cancer recurrence tends to be small when $\Delta$ is small, which results in a high variability of $Z_0^n(\gamma_n)$. The increased variability of $Z_0^n(\gamma_n)$ then affects the estimation of other parameters. We also note that as \(\alpha\) approaches 0, which indicates a high mutation rate, the estimators for \(\lambda_1\), \(\alpha\), and \(r_1\) exhibit notably poor performance. This issue can be attributed to the emergence of a large number mutant clones, resulting in a very small Simpson's Index which increases the relative error on the estimation of \(R_n\).

\subsubsection*{Performance Improvement via Bootstrapping}
Based on the stability analysis, the performance of our estimators declines when the parameter $\alpha$ approaches 1. This presents a practical challenge in implementing these estimators, especially for tumor types with low mutation rates. 

To address this issue, we propose a bootstrapping technique to enhance the accuracy in calculating the Simpson Index. In particular, when calculating Simpson's Index all bootstrap samples include the largest 20\% of clones, we then randomly sample 5/8th of the remaining clones.  After generating 1000 bootstrap samples we average the resulting Simpson's Index values to form an improved estimator. 
We examined the setting with parameters $n=10^7$, $\alpha = 0.8$, $\lambda_0 = -0.2$, and $\lambda_1 = 0.8$. In the bootstrapping process, each sample involves selecting the 20\% largest clones and resampling 50\% from the remaining clones. We conducted a comparative analysis to assess the effectiveness of the bootstrapping technique in improving estimator accuracy. By repeating the estimation process ten times and plotting the relative errors for both the estimators and the Simpson Index, we were able to compare the results before and after applying bootstrapping. As illustrated in Figure \ref{fig:bootstrapping}, the implementation of bootstrapping techniques improves the precision in calculating the Simpson Index and the associated estimators. This enhancement is particularly notable in the case of $r_1$, which typically presents a significant challenge in the estimation without bootstrapping.
\begin{figure}[h!]
    \centering
    \includegraphics[width = 0.8\textwidth]{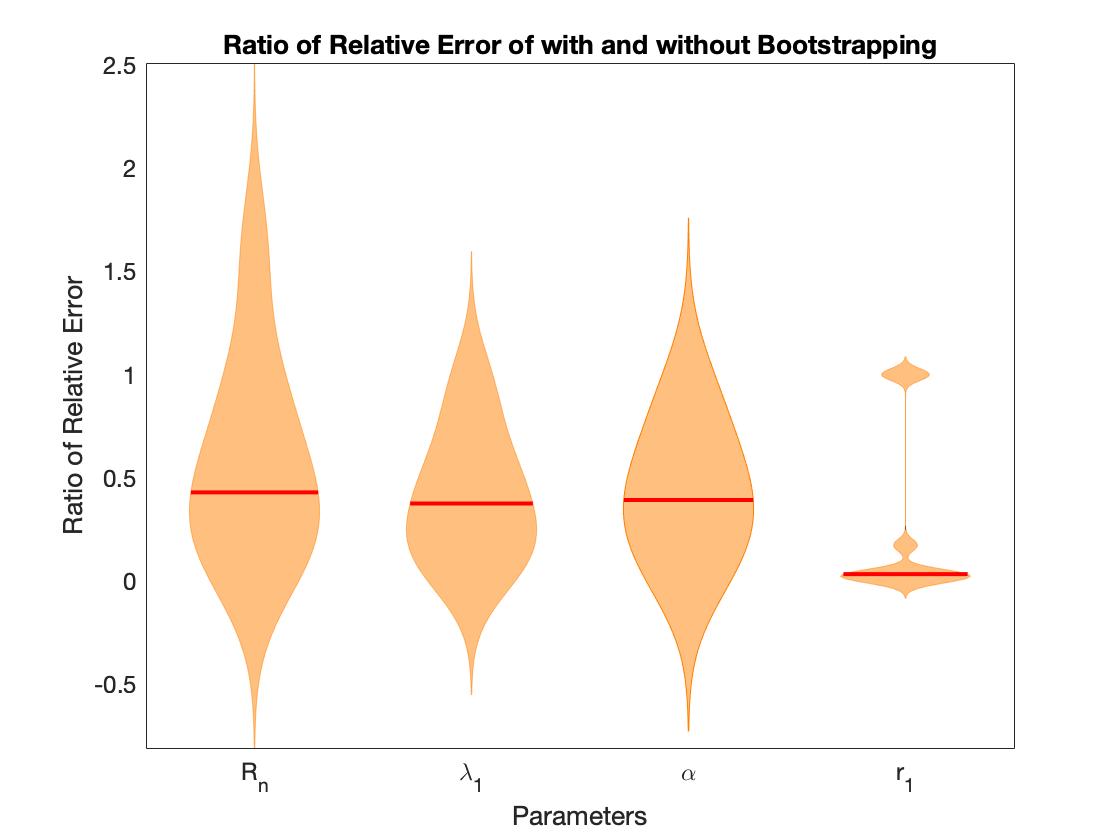}
    \caption{In this experiment, we set $n=10^7$, $\alpha = 0.8$, $\lambda_0 = -0.2$, and $\lambda_1 = 0.8$, and conduct $10$ experiments. We first derive a set of estimators using the approach developed in this paper, after which we employed bootstrapping techniques to refine these estimates. To facilitate a visual comparison of the bootstrapping technique's effectiveness, we generated violin plots illustrating the distribution of the ratio of relative errors post-bootstrapping to those pre-bootstrapping. The median of these distributions is highlighted by a red line.}
    \label{fig:bootstrapping}
\end{figure}
 
\subsection{Robustness Analysis on Carrying Capacity}

It is important to note that our model implicitly assumes an unbounded capacity for cancer cells, which is reflected through the consistent growth rate $\lambda_1$ of resistant cells. This assumption, while facilitating the analysis, might not fully capture the complexities that arise in situations where the growth of recurrent tumors is restricted by factors such as supply of nutrients or spatial constraints. Therefore, in this section, we assess the robustness of our estimators in a more practical scenario that takes into account a maximum capacity for cancer cells. Specifically, we propose a modified model where the birth rate of resistant cells is dynamically adjusted based on the total cancer cell population. Mathematically, with $C$ representing the maximum capacity for cancer cells, the adjusted birth rate of resistant cells is:
\begin{align*}
r_1(t) = r_1\left(1-\frac{Z_0^n(t)+Z_1^n(t)}{C}\right).
\end{align*}
We numerically investigate this modified model under the following parameter settings: $n=10^5$, $\alpha = 0.5$, $\lambda_0 = -0.5$, and $\lambda_1 = 0.5$. We vary the capacity parameter $C$ from $2\times 10^5$ to $5 \times 10^5$ and examine how our estimators behave and perform under these modified conditions.

\begin{figure}[h!]
\centering
   \begin{subfigure}{5cm}
    \centering\includegraphics[width=1\textwidth]{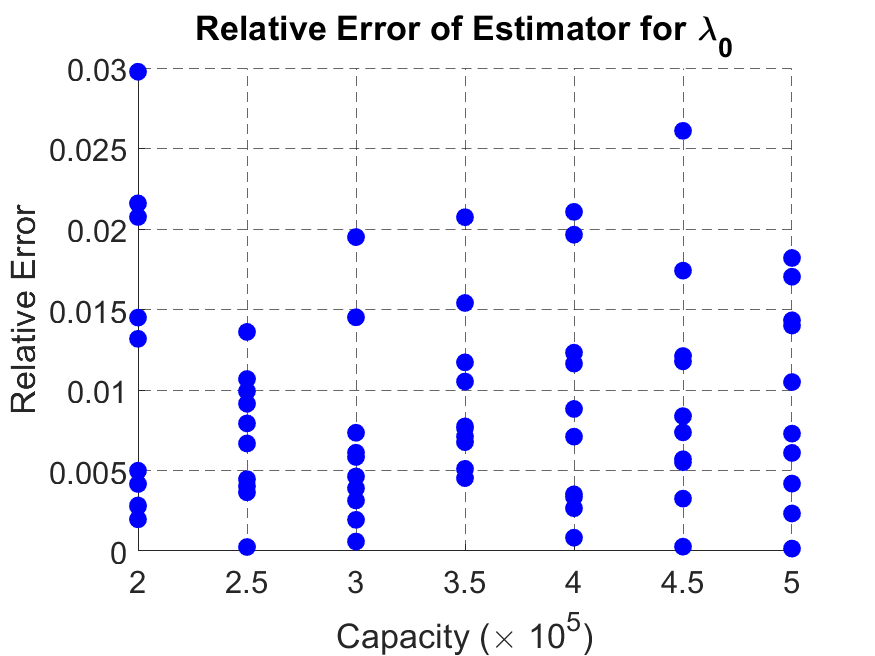}
  \end{subfigure}   
  \begin{subfigure}{5cm}
    \centering\includegraphics[width=1\textwidth]{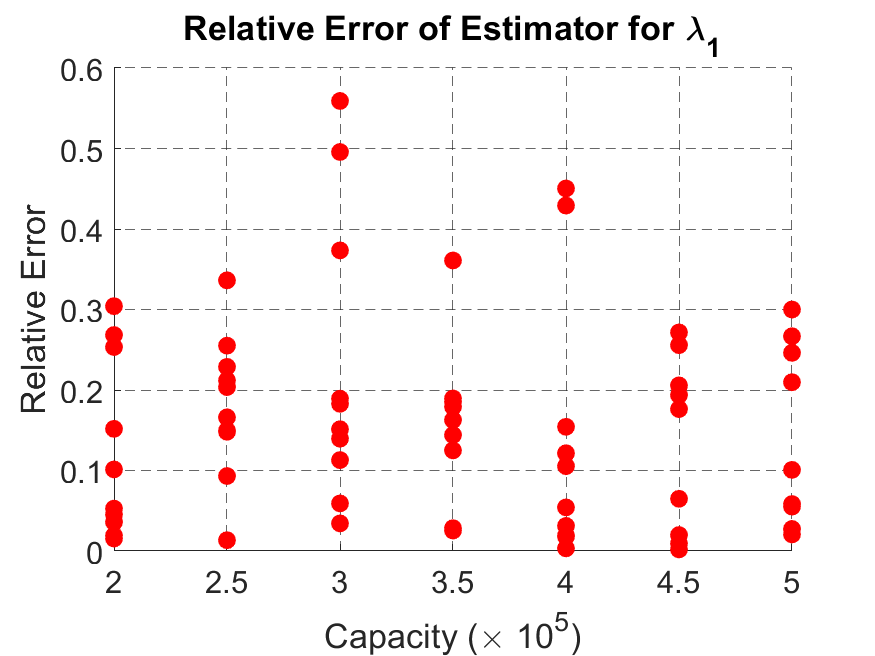}
  \end{subfigure}
  \begin{subfigure}{5cm}
    \centering\includegraphics[width=1\textwidth]{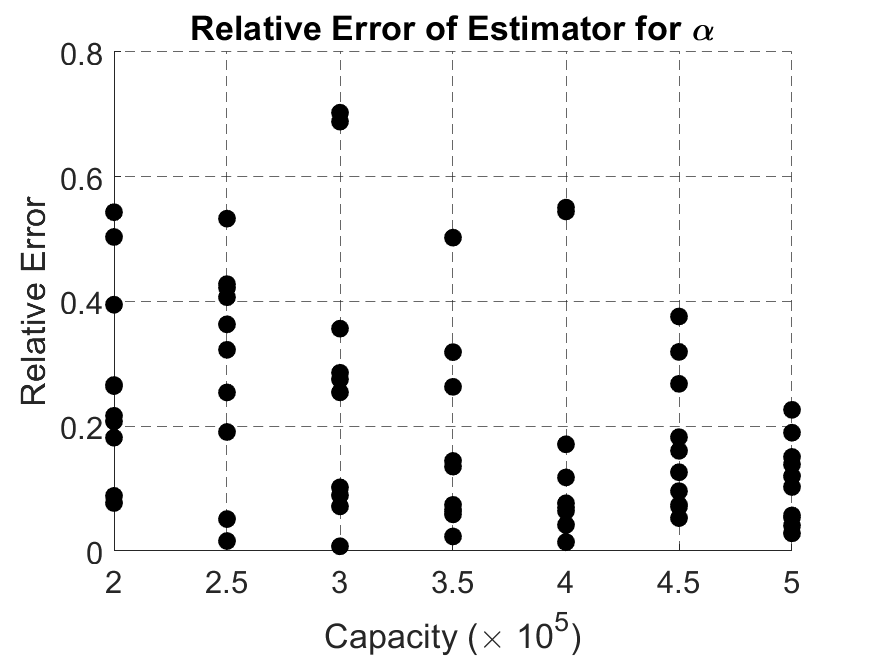}
  \end{subfigure}
\caption{In this experiment, we set $n=10^5$, $\alpha = 0.5$, $\lambda_0 = -0.5$, and $\lambda_1 = 0.5$. We vary the capacity $C$ from $2n$ to $5n$ and record the relative error of the estimators across different capacity settings.}
\label{fig:roustness_analysis}
\end{figure}

As illustrated in Figure \ref{fig:roustness_analysis}, the carrying capacity appears to have a negligible impact on the performance of our estimators for $\lambda_0$ and $\lambda_1$. However, it has a significant impact on the estimation of $\alpha$. In order to quantify these observations, we conduct a one-sided two-sample t-test to access whether a lower carrying capacity correlates with an increase in estimation error for these parameters. The p-values obtained for $\lambda_0$, $\lambda_1$, and $\alpha$ are $0.2799$, $0.5497$, and $0.0076$, respectively. These results indicate that only the estimator for $\alpha$ is significantly influenced by changes in carrying capacity.

\section{Summary and Discussion}
In summary, this study leverages a two-type birth and death model for cancer recurrence, where the cancer cells are divided into sensitive cells and resistant cells according to their fitness under pharmacological treatments. We examined several quantities that can be observed during the cancer recurrence process, namely, recurrence time, number of mutant clones, and the Simpson's Index of mutant clones. Utilizing these quantities, we developed a set of estimators for the model parameters: the growth rate of sensitive cells ($\lambda_0$), the growth rate of resistant cells ($\lambda_1$), the birth rate of resistant cells ($r_1$), and the mutation rate ($n^{-\alpha}$). These parameters are crucial for the development of personalized treatment plans for patients and in quantifying the strength of the selective pressure induced by the introduction of the anti-cancer drug or antibiotic. Through simulation, we demonstrated that these estimators exhibit robust performance across a wide range of parameter settings, underscoring their potential utility in clinical and pharmaceutical research.

In contrast to our prior work \cite{CD2021}, this paper examines a more general model wherein resistant cancer cells are modeled using a birth-death process, replacing the Yule process employed previously. A key contribution of this paper is the establishment of several convergence in probability results with explicitly specified convergence rates. These results enable the use of data from a single patient to construct estimators that are provably consistent, a significant improvement over our earlier approach that relied on convergence in expectation results and necessitated multiple data points from different patients for estimation. In practice, it's important to note that even for the same type of tumors, individual variability factors such as immune system response, genetic differences, and treatment history can affect cell fitness and mutation rates. While our previous estimators offered a general overview of these parameters, the estimators developed in this study provide individualized information, thereby facilitating more targeted medical advice. Furthermore, this paper developed an estimator for the birth rate of resistant cells ($r_1$), a variable that was not investigated in \cite{CD2021}.

This study has several limitations that could serve as potential directions for future research. Based on our simulation results, the estimators of $\lambda_0, \lambda_1$, and $\alpha$ demonstrate very good performance, whereas the estimator of $r_1$ does not exhibit equivalent performance. This observation aligns with that made by Gunnarsson and colleagues \cite{gunnarsson2023statistical}, who reported lesser accuracy in estimating the birth rate compared to the growth rate. They suggested that this discrepancy may be attributed to the fact that the birth rate can only be estimated using higher order moments. While we have established the convergence of $\hat{r}_1^{(n)}$ to $r_1$, the rate of this convergence remains undetermined. Our results suggest that, depending on the specific parameter settings, the initial burden $n$ might need to increase to as large as $10^{11}$ for error reduction to commence. Yet, our simulations indicate a trend of decline in estimation errors when the initial burden $n$ is above $10^5$. Our future work aims to establish a stronger theoretical result regarding the convergence of $r_1$.

Another limitation of our study is the assumption that all resistant cells have the same fitness level. However, this assumption is potentially oversimplified, as highlighted by \cite{elowitz2002stochastic,feinerman2008variability}, which document random gene expressions in genetically identical cancer cells. In order to consider this aspect, both \cite{JKJ2014} and \cite{Durrett2011} employ a similar multi-type branching process model, introducing heterogeneity among resistant cells by assuming random fitness effects for newly mutated cells. Expanding on this idea, an extension of our work could involve the development of estimators for models that incorporate these random fitness effects. A particularly interesting extension would be to design a non-parametric estimation framework to estimate the distribution of these random fitness changes.  

A third limitation of our current model is the assumption that drug resistance in cancer cells is driven by a single mutational event. While this assumption holds true in certain scenarios, such as the T790M mutation in non-small cell lung cancer conferring resistance to erlotinib \cite{suda2009egfr}, there are cases where multiple mutations are required for cancer cells to develop drug resistance. For example, gene amplification, which allows cancer cells to evade targeted therapies \cite{ilic2011pi3k}, necessitate a more complex model. To accommodate these scenarios, our model could be expanded to account for multiple types of cancer cells, classified based on the number of relevant mutations they possess. 

Lastly, our model does not account for the possibility of pre-existing drug resistance in cancer cells. An interesting extension to the current work would involve leveraging sequencing data obtained at the point of diagnosis and at recurrence. This approach would enable the development of refined models capable of estimating parameters that govern the dynamics of pre-existing resistant cells.

\newpage

\section{Appendix}
\subsection{{Preliminary Asymptotic Results}}
For the sake of simplicity and conciseness in the subsequent proofs, this section lists the magnitudes of several important quantities we will frequently use. For the main model, the deterministic limit of the cancer recurrence time, $\zeta_n$, and the variance of the size of resistant cells at the deterministic limit have the following orders:
\begin{align}
& \zeta_n \sim \frac{\alpha}{\lambda_1}\log n, \label{Order: zeta_n}\\
& \var\left[Z_1^n(\zeta_n)\right] \sim n^{1+\alpha}. \label{Order: second moment of Z1}
\end{align}

Denote by $Z(t)$ the population size of a birth-death process starting from a single cell with birth rate $r$, death rate $d$, and growth rate $\lambda=r-d$. We have 
\begin{align}
\label{Order: second moment Z}
    \EE\left[ Z(t)^2\right] = \frac{2r}{\lambda}e^{2\lambda t} - \frac{r+d}{\lambda}e^{\lambda t}.
\end{align}
Denote by $Z^n(t)$ the population size of a birth-death process starting from $n$ cells with the same birth/death rate specified previously. We have
\begin{align}
\label{Order: second moment Z^n}
    \EE\left[ Z^n(t)^2\right] = n\frac{r+d}{\lambda}e^{\lambda t}(e^{\lambda t}-1) +n^2e^{2\lambda t}.
\end{align}
Moreover, from \cite{JKJ2014}, we have 
\begin{align}
\label{Order: kth moment Z^n}
    \EE\left[ Z^n(t)^k\right] =\Theta\left( n^ke^{k\lambda t}\right).
\end{align}

\subsection{Proof of Theorem \ref{cip_gamma_faster}}
\label{sec:proof of cip_gamma_faster}
\begin{proof}
    We first show the following lemma, which is an extension of Lemma 2 in \cite{JK2013}.
    \begin{lemma}
    \label{lemma: Z1_converge_to_mean}
    For any $b>0$ and $u<(1-\alpha)/2$, we have 
    \begin{align*}
        \lim_{n\to\infty}\PP\left( \sup_{0\leq k \leq b}n^{\frac{\lambda_1k}{\lambda_0}-u}\left(Z_1^n(kt_n)-\Phi_1^n(kt_n)\right)>\epsilon \right) = 0,
    \end{align*}
    where $t_n = -\frac{1}{\lambda_0}\log{n}$, and $\Phi_1^n(\cdot)$ is defined in \eqref{eq:expec_Z1}.
    \end{lemma}
    \begin{proof}
        The proof of Lemma \ref{lemma: Z1_converge_to_mean} is similar to that of Lemma 2 in \cite{JK2013} with a slight change in the scaling, and thus we omit the details here. \qed
    \end{proof}

    To prove Theorem \ref{cip_gamma_faster}, we can apply the same arguments in the proof of Theorem 1 in \cite{hanagal2022large} with Lemma \ref{lemma: Z1_converge_to_mean} and a slight change in the scaling. Therefore, we omit the details here for brevity.
    
\end{proof}

\subsection{Proof of Theorem \ref{cip_In}}
\label{sec:proof of cip_In}

We first establish a convergence result for the number of mutant clones at the deterministic limit of cancer recurrence time (i.e., $\zeta_n$) in Lemma \ref{In_zeta_converge_to_limit}.
\begin{lemma}
\label{In_zeta_converge_to_limit}
    For any $\epsilon > 0, u<\min\{(1-\alpha)/2,  -\lambda_0\alpha/2\lambda_1 \}$,  we have
\begin{align}
    \label{eq:cip_In}
    \lim_{n\to\infty}\PP\left(n^u\left| \frac{1}{n^{1-\alpha}}I_n\left(\zeta_n\right)+\frac{\lambda_1}{\lambda_0 r_1}\right| > \epsilon \right) = 0.
\end{align}
\end{lemma}
\begin{proof}
Conditioned on $\left\{ Z^n_0(t), t\leq \zeta_n \right\}$, $I_n(t)$ is a non-homogeneous Poisson process with arrival rate $\frac{\lambda_1}{r_1} n^{-\alpha}Z_0^n(t)$ for $t\leq \zeta_n$. Therefore, we have
\begin{align}
\EE\left[I_n\left(\zeta_n\right)^2\right] &= \EE\left[\EE\left[I_n\left(\zeta_n\right)^2\big|Z_0^n(t), t\leq \zeta_n\right]\right] \nonumber\\
& = \EE\left[ \left(\int_0^{\zeta_n}\frac{\lambda_1}{r_1} n^{-\alpha}Z_0^n(t)dt\right)^2+ \int_0^{\zeta_n}\frac{\lambda_1}{r_1} n^{-\alpha}Z_0^n(t)dt\right] \nonumber\\
& = \frac{\lambda_1^2}{r_1^2}n^{-2\alpha}\EE\left[ \left( \int_0^{\zeta_n}Z_0^n\left(t\right)dt \right)^2 \right] +\frac{\lambda_1} {r_1}n^{1-\alpha}\int_0^{\zeta_n}e^{\lambda_0 s}ds. \label{eqn:expectation_In_square_terms}
\end{align}
For the expected value in the first term of \eqref{eqn:expectation_In_square_terms}, we have
\begin{align}
& \quad \EE\left[ \left( \int_0^{\zeta_n}Z_0^n\left(t\right)dt \right)^2 \right] \noindent\\
& \overset{\text{(a)}}{=} \int_0^{\zeta_n} \int_0^{\zeta_n} \EE\left[ Z_0^n\left(t\right) Z_0^n\left(s\right) \right]dtds \noindent\\
& = \int_0^{\zeta_n} \int_s^{\zeta_n} \EE\left[ Z_0^n\left(t\right) Z_0^n\left(s\right) \right]dtds + \int_0^{\zeta_n} \int_0^s \EE\left[ Z_0^n\left(t\right) Z_0^n\left(s\right) \right]dtds, \label{eqn:double_integral_Z_0}
\end{align}
where we use Fubini's Theorem in equality (a). For the first term in \eqref{eqn:double_integral_Z_0}, we can apply \eqref{Order: second moment Z^n} and obtain that
\begin{align*}
& \quad \int_0^{\zeta_n} \int_s^{\zeta_n} \EE\left[ Z_0^n\left(t\right) Z_0^n\left(s\right) \right]dtds\\
& = \int_0^{\zeta_n} \int_s^{\zeta_n} \EE\left[ \EE\left[ Z_0^n\left(t\right) Z_0^n\left(s\right)\vert Z_0^n(s) \right] \right]dtds\\
& = \int_0^{\zeta_n} \int_s^{\zeta_n} \EE\left[ Z_0^n\left(s\right)^2 e^{\lambda_0 \left(t-s\right)} \right]dtds\\
& = \int_0^{\zeta_n} \int_s^{\zeta_n} e^{\lambda_0\left(t-s\right)} \left(-n\frac{r_0+d_0}{\lambda_0}e^{\lambda_0 s}\left(1-e^{\lambda_0 s}\right)+n^2e^{2\lambda_0 s}\right) dtds\\
& = \int_0^{\zeta_n}  -\frac{1}{\lambda_0}\left(1-e^{\lambda_0\left(s-\zeta_n\right)}\right) \left(-n\frac{r_0+d_0}{\lambda_0}e^{\lambda_0 s}\left(1-e^{\lambda_0 s}\right)+n^2e^{2\lambda_0 s}\right) ds.
\end{align*}
Similarly, for the second term in \eqref{eqn:double_integral_Z_0}, we can obtain that
\begin{align*}
& \quad \int_0^{\zeta_n} \int_0^s \EE\left[ Z_0^n\left(t\right) Z_0^n\left(s\right) \right]dtds\\
& = \int_0^{\zeta_n} \int_0^s \EE\left[ Z_0^n\left(t\right)^2 e^{\lambda_0\left(s-t\right)} \right]dtds\\
& = \int_0^{\zeta_n} \int_0^s e^{\lambda_0\left(s-t\right)}  \left(-n\frac{r_0+d_0}{\lambda_0}e^{\lambda_0t}\left(1-e^{\lambda_0 t}\right)+n^2e^{2\lambda_0 t}\right)  dtds\\
& = \int_0^{\zeta_n} \left( -n\frac{r_0+d_0}{\lambda_0}e^{\lambda_0s}\left(s-\frac{1}{\lambda_0}e^{\lambda_0 s}+\frac{1}{\lambda_0}\right)+n^2e^{\lambda_0s}\left(\frac{1}{\lambda_0}e^{\lambda_0s}-\frac{1}{\lambda_0}\right) \right) ds.
\end{align*}
Combining these two terms, we have
\begin{align*}
& \quad \EE\left[ \left( \int_0^{\zeta_n}Z_0^n\left(t\right)dt \right)^2 \right] - \frac{n^2}{\lambda_0^2}\\
& = \int_0^{\zeta_n} n^2e^{\lambda_0 s}\left(-\frac{1}{\lambda_0}+\frac{1}{\lambda_0}e^{\lambda_0\zeta_n}\right)ds- \int_0^{\zeta_n} n\frac{r_0+d_0}{\lambda_0}e^{\lambda_0 s} \left(s-\frac{1}{\lambda_0}e^{\lambda_0\zeta_n}+\frac{1}{\lambda_0}e^{-\lambda_0\left(s-\zeta_n\right)}\right)ds -\frac{n^2}{\lambda_0^2} \\
& \sim   C_1 n^2e^{\lambda_0\zeta_n} +C_2 n,
\end{align*}
where $C_1$ and $C_2$ are some constants, and we use \eqref{Order: zeta_n} in the last step. Hence, we can get 
$$n^{2\alpha-2}\EE\left[I_n\left(\zeta_n\right)^2\right] - \frac{\lambda_1^2}{\lambda_0^2r_1^2} \sim -\frac{\lambda_1}{\lambda_0r_1}n^{\alpha-1}+ C_1 n^{\alpha\lambda_0/\lambda_1} + C_2 n^{-1}.
$$
We then use \eqref{eq: exact_In}, the Markov's Inequality, and the condition $u<\min\{(1-\alpha)/2,  -\alpha\lambda_0/2\lambda_1 \}$ to obtain that
\begin{align*}
\PP\left( n^u\left|n^{\alpha-1}I_n(\zeta_n)+\frac{\lambda_1}{\lambda_0r_1}\right|>\epsilon \right)
\leq n^{2u} \EE\left[ \left(n^{\alpha-1}I_n(\zeta_n)+\frac{\lambda_1}{\lambda_0r_1}  \right)^2 \right]/\epsilon^2\rightarrow 0.
\end{align*} \qed

Next, we show that the number of mutant clones at the cancer recurrence time (i.e., $I_n\left(\gamma_n\right)$) is close to that at the deterministic limit time (i.e., $I_n\left(\zeta_n\right)$).
\begin{lemma}
    \label{In_gamma_converge_to_In_zeta}
    For any $\epsilon > 0, u < (1-\alpha)/2-\alpha\lambda_0/\lambda_1$, we have 
    \begin{align*}
        \lim_{n\to \infty}\PP\left(n^u\left|\frac{1}{n^{1-\alpha}}\left(I_n\left(\zeta_n\right)-I_n\left(\gamma_n\right)\right)\right|>\epsilon\right) = 0.
    \end{align*}
\end{lemma}
\begin{proof}
For any $\delta > 0$, let $\delta_n = \delta n^{-k}$ where $k<(1-\alpha)/2$. We have
\begin{align*}
    \PP\left(n^u \left|\frac{1}{n^{1-\alpha}}\left(I_n\left(\zeta_n\right)-I_n\left(\gamma_n\right)\right)\right|>\epsilon \right)
\leq & \PP\left(n^u \left|\frac{1}{n^{1-\alpha}}\left(I_n\left(\zeta_n\right)-I_n\left(\gamma_n\right)\right)\right|>\epsilon, \left| \zeta_n-\gamma_n\right|< \delta_n \right)\\
&+ \PP\left( \left| \zeta_n-\gamma_n\right|> \delta_n  \right).
\end{align*}
It suffices to consider the first term as the second term goes to zero by Theorem \ref{cip_gamma_faster} and the assumption that $k<(1-\alpha)/2$. Because $I_n(t)$ is a counting process which is non-decreasing in $t$, when $\left| \zeta_n-\gamma_n\right|< \delta_n$, we have
$$
\left|\frac{1}{n^{1-\alpha}}\left(I_n\left(\zeta_n\right)-I_n\left(\gamma_n\right)\right)\right|\leq \frac{1}{n^{1-\alpha}}\left(I_n\left(\zeta_n+\delta_n\right)-I_n\left(\zeta_n-\delta_n\right)\right).
$$
Therefore, 
\begin{align*}
     & \PP\left( n^u\left|\frac{1}{n^{1-\alpha}}\left(I_n\left(\zeta_n\right)-I_n\left(\gamma_n\right)\right)\right|>\epsilon, \left| \zeta_n-\gamma_n\right|< \delta_n \right)\\
      \leq &  \PP\left( \frac{1}{n^{1-\alpha-u}}\left(I_n\left(\zeta_n+\delta_n\right)-I_n\left(\zeta_n-\delta_n\right)\right)>\epsilon\right)\\
    \leq & n^{\alpha+u-1}\EE\left[  I_n\left(\zeta_n+\delta_n\right)-I_n\left(\zeta_n-\delta_n\right)\right]/\epsilon\\
    = & \frac{\lambda_1}{\epsilon r_1}n^{u-1}\EE\left[ \int_{\zeta_n-\delta_n}^{\zeta_n+\delta_n}Z_0^n\left(t\right)dt\right]\\
    = & \frac{\lambda_1}{\epsilon r_1}n^{u-1}\int_{\zeta_n-\delta_n}^{\zeta_n+\delta_n}\EE\left[ Z_0^n\left(t\right)\right]dt\\
    =& -\frac{\lambda_1}{\epsilon \lambda_0 r_1}n^{u} \left( e^{\lambda_0(\zeta_n-\delta_n)} -e^{\lambda_0(\zeta_n+\delta_n)}\right) = \Theta(n^{\lambda_0\alpha/\lambda_1+u}\delta_n).
\end{align*}
Therefore, when $u < (1-\alpha)/2-\lambda_0\alpha/\lambda_1$, we can find $k<(1-\alpha)/2$ such that
\begin{align*}
    \limsup_{n\to\infty}\PP\left( n^u\left|\frac{1}{n^{1-\alpha}}\left(I_n\left(\zeta_n\right)-I_n\left(\gamma_n\right)\right)\right|>\epsilon, \left| \zeta_n-\gamma_n\right|< \delta_n \right) = 0,
\end{align*}
which implies that
\begin{align*}
    \lim_{n\to\infty}\PP\left( n^u\left|\frac{1}{n^{1-\alpha}}\left(I_n\left(\zeta_n\right)-I_n\left(\gamma_n\right)\right)\right|>\epsilon \right) = 0.
\end{align*} \qed
\end{proof}

Theorem \ref{cip_In} follows directly from Lemma \ref{In_zeta_converge_to_limit} and Lemma \ref{In_gamma_converge_to_In_zeta}.
\end{proof}
\subsection{Proof of Theorem \ref{cip_Rn}}
\label{sec: proof of convergence of Rn}
For ease of analysis, we redefine the Simpson's Index $R_n(t)$ by considering the number of all mutant clones generated before time $t$, which we denote by $\tilde{I}_n(t)$. The formula of the Simpson's Index becomes 
\begin{align*}
    R_n(t)=\sum_{i=1}^{\tilde{I}_n\left(t\right)}\left(\frac{X_{i,n}(t)}{Z_1^n\left(t\right)}\right)^2.
\end{align*}
This change of definition does not change the value of $R_n(t)$ as we only add some clones with size $0$ at time $t$. 

In Lemma \ref{lemma:smaller Z1}, Corollary \ref{lemma:larger Z1} and Corollary \ref{cor:Z1_larger_n}, we present a large deviation result on the size of resistant cells at a scale of cancer recurrence time.
\begin{lemma}
\label{lemma:smaller Z1}
For any $ \epsilon>0, \alpha<1$, and any $(l,k)$ such that $k<(1-\lambda_0/\lambda_1)\alpha $, $l>k$, and $l+k< \min\{1-\alpha,\alpha\}$, there exists $c>0$ so that
\begin{align*}
\limsup\limits_{n\rightarrow \infty}\frac{1}{n^{1-\alpha-(l+k)}}\log \PP\left(Z_1^n\left(\zeta_n\right)<\left(1-\epsilon_n\right)n\right)\le -c,
\end{align*}
where $\epsilon_n = \epsilon n^{-k}$.
\end{lemma}
Remark: We can always find a pair of $(l,k)$ such that the conditions specified in Lemma \ref{lemma:smaller Z1} are met. For example, we can let $l=\min\{(1-\alpha)/2,\alpha/2\}$ and choose $k<\min\{(1-\alpha)/2,\alpha/2\}$.

\begin{proof}
Similar to the arguments in the proof of Proposition 1 in \cite{CD2021}, we have 
\begin{align}
&\log\PP\left(Z_1^n\left(\zeta_n\right)<\left(1-\epsilon_n\right)n\right) \nonumber\\
\leq  & \min_{\theta_n>0} \left[ \theta_n\left(1-\epsilon_n\right)n+\frac{1}{n^{\alpha}}\int_{0}^{\zeta_n}ne^{\lambda_0 s}\left(\frac{e^{-\theta_n}-1}{\frac{r_1}{\lambda_1}\left(1-e^{-\theta_n}\right)+\frac{r_1}{\lambda_1}e^{-\theta_n}e^{-\lambda_1(\zeta_n-s)}-\frac{d_1}{\lambda_1}e^{-\lambda_1(\zeta_n-s)}}\right)ds \right]  \label{ieq:logPZ1lessthansmallern_1}\\
& \quad + k_1 (\log{n})^2 n^{1-2\alpha}, \label{ieq:logPZ1lessthansmallern_2}
\end{align}
where\eqref{ieq:logPZ1lessthansmallern_2} comes from Proposition 1 in \cite{hanagal2022large} and $k_1>0$ is some constant. 
Fix $\theta_n = \delta_n e^{-\lambda_1 \zeta_n}$ where $\delta_n = \delta n^{-l}$, and since we are now dropping the minimization operator we will have an upper bound for\eqref{ieq:logPZ1lessthansmallern_1}. From the definition of $\zeta_n$ and (\ref{eq:expec_Z1}), we have
$$
e^{-\lambda_1\zeta_n} = \frac{1}{\lambda_1-\lambda_0}n^{-\alpha}\left( 1-e^{(\lambda_0-\lambda_1)\zeta_n} \right).
$$
We can then obtain that the inner part of \eqref{ieq:logPZ1lessthansmallern_1} will be less than
\begin{align*}
&\quad \frac{\delta_n (1-\epsilon_n)}{\lambda_1-\lambda_0}n^{1-\alpha} - n^{1-\alpha}\int_{0}^{\zeta_n}e^{\lambda_0 s}\frac{\theta_n}{\frac{r_1}{\lambda_1}\theta_n+e^{-\lambda_1 \left(\zeta_n-s\right)}}ds   \\
& + n^{1-\alpha}\int_{0}^{\zeta_n}e^{\lambda_0 s}\left(\frac{e^{-\theta_n}-1}{\frac{r_1}{\lambda_1}\left(1-e^{-\theta_n}\right)+\frac{r_1}{\lambda_1}e^{-\theta_n}e^{-\lambda_1(\zeta_n-s)}-\frac{d_1}{\lambda_1}e^{-\lambda_1(\zeta_n-s)}}+\frac{\theta_n}{\frac{r_1}{\lambda_1}\theta_n+e^{-\lambda_1 \left(\zeta_n-s\right)}}\right)ds \\
&- \frac{\delta_n (1-\epsilon_n)}{\lambda_1-\lambda_0}n^{1-\alpha}e^{(\lambda_0-\lambda_1)\zeta_n} \\
\leq & \frac{\delta_n (1-\epsilon_n)}{\lambda_1-\lambda_0}n^{1-\alpha} - n^{1-\alpha}\int_{0}^{\zeta_n}e^{\lambda_0 s}\frac{\theta_n}{\frac{r_1}{\lambda_1}\theta_n+e^{-\lambda_1 \left(\zeta_n-s\right)}}ds  \quad\quad (*) \\
& + n^{1-\alpha}\int_{0}^{\zeta_n}e^{\lambda_0 s}\left(\frac{e^{-\theta_n}-1}{\frac{r_1}{\lambda_1}\left(1-e^{-\theta_n}\right)+\frac{r_1}{\lambda_1}e^{-\theta_n}e^{-\lambda_1(\zeta_n-s)}-\frac{d_1}{\lambda_1}e^{-\lambda_1(\zeta_n-s)}}+\frac{\theta_n}{\frac{r_1}{\lambda_1}\theta_n+e^{-\lambda_1 \left(\zeta_n-s\right)}}\right)ds.  \quad\quad (**)
\end{align*}
For the first term, we have 
\begin{align*}
    (*) &= \frac{\delta_n (1-\epsilon_n)}{\lambda_1-\lambda_0}n^{1-\alpha} - n^{1-\alpha}\int_{0}^{\zeta_n}e^{\lambda_0 s}\frac{\delta_n}{\frac{r_1}{\lambda_1}\delta_n+e^{\lambda_1 s}}ds\\
    &= n^{1-\alpha-l}\left( \frac{\delta}{\lambda_1-\lambda_0}(1-\epsilon_n)-\int_{0}^{\zeta_n}e^{\lambda_0 s}\frac{\delta}{\frac{r_1}{\lambda_1}\delta_n+e^{\lambda_1 s}}ds \right)\\
    &= \delta n^{1-\alpha-l}\left( -\frac{\epsilon_n}{\lambda_1-\lambda_0}+\int_0^{\infty} \frac{e^{\lambda_0 s}}{e^{\lambda_1 s}}ds-\int_{0}^{\zeta_n}\frac{e^{\lambda_0 s}}{\frac{r_1}{\lambda_1}\delta_n+e^{\lambda_1 s}}ds \right)\\
    &= \delta n^{1-\alpha-l}\left( -\frac{\epsilon_n}{\lambda_1-\lambda_0}+\frac{e^{(\lambda_0-\lambda_1)\zeta_n}}{\lambda_1-\lambda_0}+\int_{0}^{\zeta_n}\left(\frac{e^{\lambda_0 s}}{e^{\lambda_1 s}}-\frac{e^{\lambda_0 s}}{\frac{r_1}{\lambda_1}\delta_n+e^{\lambda_1 s}}\right)ds \right)\\
    &= \delta n^{1-\alpha-l}\left(-\frac{\epsilon_n}{\lambda_1-\lambda_0}+\frac{e^{(\lambda_0-\lambda_1)\zeta_n}}{\lambda_1-\lambda_0} +\frac{r_1}{\lambda_1}\delta_n\int_0^{\zeta_n} \frac{e^{(\lambda_0-\lambda_1) s}}{\frac{r_1}{\lambda_1}\delta_n+e^{\lambda_1 s}}ds \right).
\end{align*}
By \eqref{Order: zeta_n}, we have
\begin{align*}
   e^{(\lambda_0-\lambda_1)\zeta_n}= \Theta(n^{-(1-\lambda_0/\lambda_1)\alpha}), \text{ } \delta_n= \Theta(n^{-l}), \text{ and } \epsilon_n= \Theta(n^{-k}).
\end{align*}
Therefore, when $l>k$ and $k<(1-\lambda_0/\lambda_1)\alpha$, $(*)$ will be dominated by the  term $-\frac{\delta\epsilon_n}{\lambda_1-\lambda_0}n^{1-\alpha-l}$ and thus $(*) = \Theta(n^{1-\alpha-(l+k)})$ and is negative. Meanwhile, for the second term, we have
\begin{align}
    (**) =   n^{1-\alpha}\int_{0}^{\zeta_n}e^{\lambda_0 s}\frac{\left(\frac{r_1}{\lambda_1}\theta_n e^{-\theta_n} -1-\frac{d_1}{\lambda_1}\theta_n + e^{-\theta_n}  \right)e^{-\lambda_1(\zeta_n-s)}}{\left(\frac{r_1}{\lambda_1}\theta_n+e^{-\lambda_1 \left(\zeta_n-s\right)}\right)\left(\frac{r_1}{\lambda_1}\left(1-e^{-\theta_n}\right)+\frac{r_1}{\lambda_1}e^{-\theta_n}e^{-\lambda_1 \left(\zeta_n-s\right)}-\frac{d_1}{\lambda_1}e^{-\lambda_1(\zeta_n-s)}\right)}ds. \label{eqn_two_star_exp}
\end{align}
Because $r_1>d_1$ and $\theta_n$ converges to 0, by the Taylor expansion of $e^{-\theta_n}$, we can obtain that the numerator in \eqref{eqn_two_star_exp} is negative and the denominator is positive when $n$ is sufficiently large. Therefore we can get an upper bound of the absolute value of $(**)$ as follows.
\begin{align*}
    |(**)| &\leq n^{1-\alpha}\int_{0}^{\zeta_n}e^{\lambda_0 s}\frac{\left(\left(\frac{r_1}{\lambda_1}-\frac{1}{2}\right)\theta_n^2+\frac{1}{6}\theta_n^3\right) e^{-\lambda_1(\zeta_n-s)}}{\left(\frac{r_1}{\lambda_1}\theta_n+e^{-\lambda_1 \left(\zeta_n-s\right)}\right)\left(\frac{r_1}{\lambda_1}\left(\theta_n-\theta_n^2/2\right)+\frac{r_1}{\lambda_1}(1-\theta_n)e^{-\lambda_1 \left(\zeta_n-s\right)}-\frac{d_1}{\lambda_1}e^{-\lambda_1(\zeta_n-s)}\right)}ds\\
    &\sim \left(\frac{r_1}{\lambda_1} -\frac{1}{2}\right)n^{1-\alpha}\delta_n^2e^{-\lambda_1\zeta_n}\int_0^{\zeta_n}\frac{e^{(\lambda_0+\lambda_1)s}}{e^{2\lambda_1s}}ds = \Theta(n^{1-2\alpha-2l}).
\end{align*}
Because $l>k$ and $\alpha > 0$, we have $1-2\alpha-2l< 1-\alpha-(l+k)$ which implies that $(**)$ is dominated by $(*)$. Finally, because $l+k<\alpha$, \eqref{ieq:logPZ1lessthansmallern_2} is dominated by $(*)$ as well, which leads to the desired result such that for some $c>0$, we have
\begin{align*}
\limsup\limits_{n\rightarrow \infty}\frac{1}{n^{1-\alpha-(l+k)}}\log \PP\left(Z_1^n\left(\zeta_n\right)<\left(1-\epsilon_n\right)n\right)\le -c.
\end{align*} \qed
\end{proof}

\begin{corollary}
\label{lemma:larger Z1}
For any $ \epsilon>0, \alpha<1$, and any $(l,k)$ such that $k<(1+r/\lambda_1)\alpha $, $l>k$ and $l+k< \min\{1-\alpha,\alpha\}$, there exists $c>0$ so that
\begin{align*}
\limsup\limits_{n\rightarrow \infty}\frac{1}{n^{1-\alpha-(l+k)}}\log \PP\left(Z_1^n\left(\zeta_n\right)>\left(1+\epsilon_n\right)n\right)\le -c,
\end{align*}
where $\epsilon_n = \epsilon n^{-k}$.
\end{corollary}
\begin{proof}
Similar to the proof of Lemma \ref{lemma:smaller Z1}, we can obtain that
\begin{align}
&\log\PP\left(Z_1^n\left(\zeta_n\right)>\left(1+\epsilon_n\right)n\right) \nonumber \\
\leq  & \min_{0<\theta_n}\left[ -\theta_n\left(1+\epsilon_n\right)n+n^{1-\alpha}\int_{0}^{\zeta_n}e^{\lambda_0 s}\left(\frac{e^{\theta_n}-1}{\frac{r_1}{\lambda_1}\left(1-e^{\theta_n}\right)+\frac{r_1}{\lambda_1}e^{\theta_n}e^{-\lambda_1(\zeta_n-s)}-\frac{d_1}{\lambda_1}e^{-\lambda_1(\zeta_n-s)}}\right)ds \right]\label{ieq:logPZ1largerthanbiggern_1}\\
& \quad + k_1 (\log{n})^2 n^{1-2\alpha}. \label{ieq:logPZ1largerthanbiggern_2}
\end{align}
Let $\theta_n = \delta_n e^{-\lambda_1 \zeta_n}$ where $\delta_n = \delta n^{-l}$. We can obtain that the inner part of \eqref{ieq:logPZ1largerthanbiggern_1} is equal to
\begin{align*}
&-\theta_n\left(1+\epsilon_n\right)n+n^{1-\alpha}\int_{0}^{\zeta_n}e^{\lambda_0 s}\left(\frac{e^{\theta_n}-1}{\frac{r_1}{\lambda_1}\left(1-e^{\theta_n}\right)+\frac{r_1}{\lambda_1}e^{\theta_n}e^{-\lambda_1(\zeta_n-s)}-\frac{d_1}{\lambda_1}e^{-\lambda_1(\zeta_n-s)}}\right)ds \\
&= -\frac{\delta_n (1+\epsilon_n)}{\lambda_1-\lambda_0}n^{1-\alpha} + n^{1-\alpha}\int_{0}^{\zeta_n}e^{\lambda_0 s}\frac{\delta_n}{-\frac{r_1}{\lambda_1}\delta_n+e^{\lambda_1 s}}ds   \quad\quad (*) \\
& \quad + n^ {1-\alpha}\int_{0}^{\zeta_n}e^{\lambda_0 s}\left(\frac{e^{\theta_n}-1}{\frac{r_1}{\lambda_1}\left(1-e^{\theta_n}\right)+\frac{r_1}{\lambda_1}e^{\theta_n}e^{-\lambda_1(\zeta_n-s)}-\frac{d_1}{\lambda_1}e^{-\lambda_1(\zeta_n-s)}}-\frac{\delta_n}{-\frac{r_1}{\lambda_1}\delta_n+e^{\lambda_1 s}}\right)ds \quad\quad (**) \\
& \quad +\frac{\delta_n (1+\epsilon_n)}{\lambda_1-\lambda_0}n^{1-\alpha}e^{(\lambda_0-\lambda_1)\zeta_n}.  \quad\quad (***)
\end{align*}
For the first term, we have 
\begin{align*}
    (*) &= -\frac{\delta_n (1+\epsilon_n)}{\lambda_1-\lambda_0}n^{1-\alpha} + n^{1-\alpha}\int_{0}^{\zeta_n}e^{\lambda_0 s}\frac{\delta_n}{-\frac{r_1}{\lambda_1}\delta_n+e^{\lambda_1 s}}ds\\
    &= n^{1-\alpha-l}\left( -\frac{\delta}{\lambda_1-\lambda_0}(1+\epsilon_n)+\int_{0}^{\zeta_n}e^{\lambda_0 s}\frac{\delta}{-\frac{r_1}{\lambda_1}\delta_n+e^{\lambda_1 s}}ds \right)\\
    &= \delta n^{1-\alpha-l}\left( -\frac{\epsilon_n}{\lambda_1-\lambda_0}-\int_0^{\infty} \frac{e^{\lambda_0 s}}{e^{\lambda_1 s}}ds+\int_{0}^{\zeta_n}\frac{e^{\lambda_0 s}}{-\frac{r_1}{\lambda_1}\delta_n+e^{\lambda_1 s}}ds \right)\\
    &= \delta n^{1-\alpha-l}\left( -\frac{\epsilon_n}{\lambda_1-\lambda_0}-\frac{e^{(\lambda_0-\lambda_1)\zeta_n}}{\lambda_1-\lambda_0}-\int_{0}^{\zeta_n}\left(\frac{e^{\lambda_0 s}}{e^{\lambda_1 s}}-\frac{e^{\lambda_0 s}}{-\frac{r_1}{\lambda_1}\delta_n+e^{\lambda_1 s}}\right)ds \right)\\
    &= \delta n^{1-\alpha-l}\left(-\frac{\epsilon_n}{\lambda_1-\lambda_0}-\frac{e^{(\lambda_0-\lambda_1)\zeta_n}}{\lambda_1-\lambda_0} -\frac{r_1}{\lambda_1}\delta_n\int_0^{\zeta_n} \frac{e^{(\lambda_0-\lambda_1) s}}{-\frac{r_1}{\lambda_1}\delta_n+e^{\lambda_1 s}}ds \right).
\end{align*}
When $l>k$ and $k<(1+r/\lambda_1)\alpha$, we have $(*) = \Theta(n^{1-\alpha-l-k})$. The analysis of the second term and the third term are similar to that in Lemma \ref{lemma:smaller Z1} and we can obtain that
\begin{align*}
    &(**) = \Theta(n^{1-2\alpha-2l})\\
    &(***) = \Theta(n^{1-\alpha-l-(1+r/\lambda_1)\alpha}),
\end{align*}
which are both dominated by the first term. When $l+k<\alpha$, the \eqref{ieq:logPZ1largerthanbiggern_2} is also dominated by $(*)$. Therefore, for some $c>0$, we have
\begin{align*}
\limsup\limits_{n\rightarrow \infty}\frac{1}{n^{1-\alpha-(l+k)}}\log \PP\left(Z_1^n\left(\zeta_n\right)>\left(1+\epsilon_n\right)n\right)\le -c. 
\end{align*} \qed
\end{proof}

\begin{corollary}
\label{cor:Z1_larger_n}
For any $ \epsilon>0, \alpha<1$ and  $(l,k)$ satisfying $1<k<l<k+\frac{1-\alpha}{\alpha}   $, there exists $c>0$ so that
\begin{align*}
\limsup\limits_{n\rightarrow \infty}\frac{1}{n^{1-(l-k+1)\alpha}}\log \PP\left(Z_1^n\left(k\zeta_n\right)<n\right)\le -c,
\end{align*}
\end{corollary}
\begin{proof}
    Similar with the proof of the Lemma \ref{lemma:smaller Z1}, we have 
\begin{align}
&\log\PP\left(Z_1^n\left(\zeta_n\right)<\left(1-\epsilon_n\right)n\right) \nonumber\\
\leq  & \min_{\theta_n>0} \left[ \theta_nn+\frac{1}{n^{\alpha}}\int_{0}^{k\zeta_n}ne^{\lambda_0 s}\left(\frac{e^{-\theta_n}-1}{\frac{r_1}{\lambda_1}\left(1-e^{-\theta_n}\right)+\frac{r_1}{\lambda_1}e^{-\theta_n}e^{-\lambda_1(k\zeta_n-s)}-\frac{d_1}{\lambda_1}e^{-\lambda_1(k\zeta_n-s)}}\right)ds \right] \\
& \quad + k_1 (\log{n})^2 n^{1-2\alpha}, 
\end{align}
Let $\theta_n = \delta e^{-l\lambda_1\zeta_n}$ where $l$ can be any number satisfying  $k<l<k+\frac{1-\alpha}{\alpha}$. We can obtain
\begin{align*}
    &n^{(l-k+1)\alpha-1}\left( \theta_nn+\frac{1}{n^{\alpha}}\int_{0}^{k\zeta_n}ne^{\lambda_0 s}\left(\frac{e^{-\theta_n}-1}{\frac{r_1}{\lambda_1}\left(1-e^{-\theta_n}\right)+\frac{r_1}{\lambda_1}e^{-\theta_n}e^{-\lambda_1(k\zeta_n-s)}-\frac{d_1}{\lambda_1}e^{-\lambda_1(k\zeta_n-s)}}\right)ds \right)\\
    =& \delta n^{(l-k+1)\alpha}e^{-l\lambda_1\zeta_n}+n^{(l-k)\alpha}\int_{0}^{k\zeta_n}e^{\lambda_0 s}\left(\frac{e^{-\theta_n}-1}{\frac{r_1}{\lambda_1}\left(1-e^{-\theta_n}\right)+\frac{r_1}{\lambda_1}e^{-\theta_n}e^{-\lambda_1(k\zeta_n-s)}-\frac{d_1}{\lambda_1}e^{-\lambda_1(k\zeta_n-s)}}\right)ds\\
    =& \delta n^{(l-k+1)\alpha}e^{-l\lambda_1\zeta_n} - n^{(l-k)\alpha}e^{k\lambda_0\zeta_n}\int_{0}^{k\zeta_n}e^{(\lambda_1-\lambda_0) s}\left(\frac{1-e^{-\theta_n}}{\frac{r_1}{\lambda_1}\left(1-e^{-\theta_n}\right)e^{\lambda_1s}+\frac{r_1}{\lambda_1}e^{-\theta_n}-\frac{d_1}{\lambda_1}}\right)ds\\
    =& \delta n^{(l-k+1)\alpha}e^{-l\lambda_1\zeta_n} - n^{(l-k)\alpha}e^{k\lambda_0\zeta_n}\left(1-e^{-\theta_n}\right)\int_{0}^{k\zeta_n}\left(\frac{e^{(\lambda_1-\lambda_0) (k\zeta_n-s)}}{\frac{r_1}{\lambda_1}\left(1-e^{-\theta_n}\right)e^{\lambda_1(k\zeta_n-s)}+\frac{r_1}{\lambda_1}e^{-\theta_n}-\frac{d_1}{\lambda_1}}\right)ds\\
    =& \delta n^{(l-k+1)\alpha}e^{-l\lambda_1\zeta_n} - n^{(l-k)\alpha}e^{k\lambda_0\zeta_n}\left(1-e^{-\theta_n}\right)e^{k(\lambda_1-\lambda_0) \zeta_n}\int_{0}^{k\zeta_n}\left(\frac{e^{-(\lambda_1-\lambda_0) s}}{\frac{r_1}{\lambda_1}\left(1-e^{-\theta_n}\right)e^{\lambda_1(k\zeta_n-s)}+\frac{r_1}{\lambda_1}e^{-\theta_n}-\frac{d_1}{\lambda_1}}\right)ds\\
\end{align*}
By Dominated Convergence Theorem, we have 
\begin{align*}
    \lim_{n\to\infty}\int_{0}^{k\zeta_n}\left(\frac{e^{-(\lambda_1-\lambda_0) s}}{\frac{r_1}{\lambda_1}\left(1-e^{-\theta_n}\right)e^{\lambda_1(k\zeta_n-s)}+\frac{r_1}{\lambda_1}e^{-\theta_n}-\frac{d_1}{\lambda_1}}\right)ds  = \int_0^{\infty} e^{-(\lambda_1-\lambda_0) s}ds = \frac{1}{\lambda_1-\lambda_0}
\end{align*}
Because of $1-e^{-\theta_n}\sim \theta_n =\delta n^{-l\alpha} $ and $e^{k\lambda_1\zeta_n}\sim (\lambda_1-\lambda_0)^k n^{k\alpha}$, we can find $c>0$ and 
\begin{align*}
   \limsup\limits_{n\rightarrow \infty} n^{(l-k+1)\alpha-1}\left( \theta_nn+\frac{1}{n^{\alpha}}\int_{0}^{k\zeta_n}ne^{\lambda_0 s}\left(\frac{e^{-\theta_n}-1}{\frac{r_1}{\lambda_1}\left(1-e^{-\theta_n}\right)+\frac{r_1}{\lambda_1}e^{-\theta_n}e^{-\lambda_1(k\zeta_n-s)}-\frac{d_1}{\lambda_1}e^{-\lambda_1(k\zeta_n-s)}}\right)ds \right)< -c 
\end{align*} 
\qed
\end{proof}

In Lemma \ref{lemma:Convergence_of_Rn_zeta_n}, we show a convergence in probability result for the Simpson's Index of mutant clones at the deterministic limit of cancer recurrence time. For ease of exposition, we define $\PP_{\rho_n}(A) = \PP(A\cap \rho_n)$ and $\EE_{\rho_n}\left[ X \right] = \EE\left[ X\cdot 1_{\rho_n} \right]$.
\begin{lemma}
\label{lemma:Convergence_of_Rn_zeta_n}
For any $\epsilon > 0$, when $u<\min\{\alpha/4,(1-\alpha)/4\}$, we have 
\begin{align*}
    \lim_{n\to\infty}\PP_{\rho_n}\left( n^u\left| n^{1-\alpha}R_n\left(\zeta_n\right)- \frac{2r_1\left(\lambda_1-\lambda_0\right)^2}{ \lambda_1 \left(2\lambda_1-\lambda_0\right)} \right| > \epsilon\right) = 0.
\end{align*}
\end{lemma}
\begin{proof}
By Chebyshev's inequality, we have
\begin{align*}
    &\PP_{\rho_n}\left( n^u\left| n^{1-\alpha}R_n\left(\zeta_n\right)- \frac{2r_1\left(\lambda_1-\lambda_0\right)^2}{ \lambda_1 \left(2\lambda_1-\lambda_0\right)} \right| > \epsilon\right) \\
    \leq &\frac{n^{2-2\alpha+2u}}{\epsilon^2}\EE_{\rho_n}\left[\left(R_n(\zeta_n)-n^{\alpha-1}\frac{2r_1\left(\lambda_1-\lambda_0\right)^2}{ \lambda_1 \left(2\lambda_1-\lambda_0\right)}\right)^2\right]  \\
    \leq & \frac{n^{2-2\alpha+2u}}{\epsilon^2}\EE\left[\left(R_n(\zeta_n)-n^{\alpha-1}\frac{2r_1\left(\lambda_1-\lambda_0\right)^2}{ \lambda_1 \left(2\lambda_1-\lambda_0\right)}\right)^2\right].  
\end{align*}
   Combined with \ref{ciE_Simp}, it suffices to show that
\begin{align*}
    &\lim_{n\to \infty} \frac{n^{2-2\alpha+2u}}{\epsilon^2}\left( \EE\left[ R_n(\zeta_n)^2 \right] - n^{2\alpha-2}\frac{4r_1^2\left(\lambda_1-\lambda_0\right)^4}{ \lambda_1^2 \left(2\lambda_1-\lambda_0\right)^2} \right) = 0.
\end{align*}
Define
\begin{align}
\tilde{R}_n(\zeta_n)=\sum_{i=1}^{\tilde{I}_n\left(\zeta_n\right)}\left(\frac{X_{i,n}(\zeta_n)}{n}\right)^2.
\end{align}
We can obtain that 
\begin{align*}
     n^{2u}\left|\EE\left[ R_n(\zeta_n)^2 \right]-\EE\left[ \tilde{R}_n(\zeta_n)^2 \right]\right|  &\leq  n^{2u}\EE\left[ \left|R_n(\zeta_n)^2-\tilde{R}_n(\zeta_n)^2\right| \right]  \\
      &=  n^{2u}\EE_{\rho_n}\left[ \left|R_n(\zeta_n)^2-\tilde{R}_n(\zeta_n)^2\right| \right]  \\
     &\leq n^{2u}\EE_{\rho_n}\left[\left|\tilde{R}_n(\zeta_n)^2\left(\frac{n^4}{\left(Z_1^n\left(\zeta_n\right)\right)^4}-1\right)\right|\right].
\end{align*}
From Lemma \ref{lemma:smaller Z1} and Corollary \ref{lemma:larger Z1}, we know $\PP_{\rho_n}\left(Z_1^n\left(\zeta_n\right)<\left(1-\delta_n\right)n\right)\leq \PP\left(Z_1^n\left(\zeta_n\right)<\left(1-\delta_n\right)n\right)$ and $\PP_{\rho_n}\left(Z_1^n\left(\zeta_n\right)>\left(1+\delta_n\right)n\right)\leq \PP\left(Z_1^n\left(\zeta_n\right)>\left(1+\delta_n\right)n\right)$
decays exponentially fast for any $\delta_n = \delta n^{-k}$ where $k<\min\{(1-\alpha)/2, \alpha/2\}$. Therefore, 
\begin{align*}
    & \quad n^{2u}\EE_{\rho_n}\left[\left|\tilde{R}_n(\zeta_n)^2\left(\frac{n^4}{\left(Z_1^n\left(\zeta_n\right)\right)^4}-1\right)\right|\right]\\
    & = n^{2u}\EE_{\rho_n}\left[\left|\tilde{R}_n(\zeta_n)^2\left(\frac{n^4}{\left(Z_1^n\left(\zeta_n\right)\right)^4}-1\right)\right|, Z_1^n\left(\zeta_n\right)<\left(1-\delta_n\right)n\right]\\
    &\quad \quad + n^{2u}\EE_{\rho_n}\left[\left|\tilde{R}_n(\zeta_n)^2\left(\frac{n^4}{\left(Z_1^n\left(\zeta_n\right)\right)^4}-1\right)\right|, Z_1^n\left(\zeta_n\right)\ge \left(1+\delta_n\right)n\right]\\
    & \quad \quad + n^{2u}\EE_{\rho_n}\left[\left|\tilde{R}_n(\zeta_n)^2\left(\frac{n^4}{\left(Z_1^n\left(\zeta_n\right)\right)^4}-1\right)\right|,(1+\delta_n)n> Z_1^n\left(\zeta_n\right)\ge \left(1-\delta_n\right)n\right]\\
    & \leq n^{4+2u}\EE_{\rho_n}\left[\tilde{R}_n(\zeta_n)^2\vert Z_1^n\left(\zeta_n\right)<\left(1-\delta_n\right)n\right]\PP_{\rho_n}\left(Z_1^n\left(\zeta_n\right)<\left(1-\delta_n\right)n\right)\\
    &\quad\quad + n^{2u}\EE_{\rho_n}\left[\frac{Z_1^n(\zeta_n)^4}{n^4},Z_1^n\left(\zeta_n\right)\ge \left(1+\delta_n\right)n\right]\\
    &\quad\quad + n^{2u}\left(1/(1-\delta_n)^4 -1 \right) \EE_{\rho_n}\left[\tilde{R}_n(\zeta_n)^2\right]\\
    & \leq n^{4+2u}\PP_{\rho_n}\left(Z_1^n\left(\zeta_n\right)<\left(1-\delta_n\right)n\right)+ n^{2u}\left(1/(1-\delta_n)^4 -1 \right) \EE_{\rho_n}\left[\tilde{R}_n(\zeta_n)^2\right]\\
    \quad\quad &+ n^{2u-4}\EE_{\rho_n}\left[Z_1^n(\zeta_n)^8\right]^{1/2}\PP_{\rho_n}\left(Z_1^n\left(\zeta_n\right)\ge \left(1+\delta_n\right)n\right)^{1/2}  \quad\quad \quad \mbox{ Cauthy-Schwarz Inequality  }\\
    & \overset{\text{(a)}}{\sim}  n^{2u}\left(1/(1-\delta_n)^4 -1 \right) \EE_{\rho_n}\left[\tilde{R}_n(\zeta_n)^2\right]\\
    & \sim 4n^{2u}\delta_n \EE_{\rho_n}\left[\tilde{R}_n(\zeta_n)^2\right],
\end{align*}
where we use the following fact in (a):
\begin{equation}
\EE_{\rho_n}\left[Z_1^n(\zeta_n)^8\right]\le C n^8, \text{ for some constant $C>0$}. \label{eqn_bound_on_eighth_moment_Z_1}
\end{equation}
To prove \eqref{eqn_bound_on_eighth_moment_Z_1}, we first observe that
\begin{align}
    \EE_{\rho_n}\left[Z_1^n(\zeta_n)^8\right] & =  \EE\left[Z_1^n(\zeta_n)^8\right] \nonumber\\
    & \overset{\text{(a)}}{\leq} \EE\left[\tilde{I}_n(\zeta_n)^8Z_1(\zeta_n)^8\right] \nonumber\\
    & = \EE\left[\tilde{I}_n(\zeta_n)^8\right]\EE\left[Z_1(\zeta_n)^8\right], \label{eqn:I_n_and_Z_1_zeta_n}
\end{align}
where $Z_1(t)$ is the population size of a branching process starting from a single resistant cell with birth rate $r_1$ and death rate $d_1$, and the inequality (a) holds because $Z_1^n(\zeta_n)$ is bounded by the process assuming all the clones are generated at $t =0$ and remember $\tilde{I}_n(\zeta_n)$ represents the number of all mutant clones generated before time $\zeta_n$. For the first term in \eqref{eqn:I_n_and_Z_1_zeta_n}, we have
\begin{align*}
    \EE\left[\tilde{I}_n(\zeta_n)^8\right] &= \EE\left[\EE\left[\tilde{I}_n(\zeta_n)^8\vert Z^n_0(t), t \leq \zeta_n\right]\right]\\
    & \overset{\text{(a)}}{=} \Theta\left( \EE\left[ \left(\int_0^{\zeta_n}n^{-\alpha}Z_0^n(t)dt\right)^8\right]\right),
\end{align*}
where we use the fact that the eighth moment of a Poisson distribution has the same order as the eighth power of its mean in equality (a). We can further show that
\begin{align*}
    &\EE\left[ \left(\int_0^{\zeta_n} n^{-\alpha}Z_0^n(t)dt\right)^8\right] \\
    = & n^{-8\alpha}\int_0^{\zeta_n} \cdots\int_0^{\zeta_n} \EE\left[ Z_0^n(x_1)\cdots Z_0^n(x_8) \right]  dx_1\cdots dx_8\\
    \overset{\text{(a)}}{\leq} & n^{-8\alpha}\left( \int_0^{\zeta_n} \EE\left[ Z_0^n(s)^8 \right]^{1/8}  ds\right)^8 \\
    = &  \Theta(n^{8-8\alpha}),  
\end{align*}
where we use Cauchy-Schwarz Inequality repeatedly in inequality (a). As can be seen in Lemma 5 from \cite{JKJ2014}, we have $\EE\left[Z_1(\zeta_n)^8\right]=\Theta(n^{8\alpha})$, which completes the proof of \eqref{eqn_bound_on_eighth_moment_Z_1}. To conclude, we have obtained that 
\begin{align}
n^{2u}\left|\EE\left[ R_n(\zeta_n)^2 \right]-\EE\left[ \tilde{R}_n(\zeta_n)^2 \right]\right|=\Theta(4n^{2u}\delta_n \EE_{\rho_n}\left[\tilde{R}_n(\zeta_n)^2\right]). \label{eqn:lemma_5_component_1}
\end{align} 
From the proof of Proposition 1 in \cite{CD2021}, we can obtain that
\begin{align}
    \EE_{\rho_n}\left[\tilde{R}_n(\zeta_n)^2\right]  =   \frac{4r_1^2\left(\lambda_1-\lambda_0\right)^4}{\lambda_1^2 \left(2\lambda_1-\lambda_0\right)^2}n^{-2+2\alpha}+C_1 n^{-2+\alpha} + C_2n^{-3+3\alpha} + o(n^{-2+\alpha}), \label{eqn:lemma_5_component_2}
\end{align}
where $C_1$ and $C_2$ are some constants.
With \eqref{eqn:lemma_5_component_1} and \eqref{eqn:lemma_5_component_2}, we have 
\begin{align*}
     &n^{2-2\alpha+2u}\left( \left|\EE\left[ R_n(\zeta_n)^2 \right] - n^{2\alpha-2}\frac{4r_1^2\left(\lambda_1-\lambda_0\right)^4}{ \lambda_1^2 \left(2\lambda_1-\lambda_0\right)^2} \right|\right)\\
     \leq & n^{2-2\alpha+2u} \left(\left|\EE\left[ R_n(\zeta_n)^2 \right]-\EE\left[ \tilde{R}_n(\zeta_n)^2 \right]\right| + \left| \EE\left[ \tilde{R}_n(\zeta_n)^2 \right] - n^{2\alpha-2}\frac{4r_1^2\left(\lambda_1-\lambda_0\right)^4}{ \lambda_1^2 \left(2\lambda_1-\lambda_0\right)^2}  \right|\right)\\
     = & \Theta\left(  4n^{2u}\delta_n n^{2-2\alpha}\EE_{\rho_n}\left[\tilde{R}_n(\zeta_n)^2\right] + n^{2u-\alpha} + n^{2u-1+\alpha}\right).
\end{align*}
Therefore, when $u<\min\{\alpha/4,(1-\alpha)/4\} $, we have 
$$
\lim_{n\to\infty}\PP_{\rho_n}\left( n^u\left| n^{1-\alpha}R_n\left(\zeta_n\right)- \frac{2r_1\left(\lambda_1-\lambda_0\right)^2}{ \lambda_1 \left(2\lambda_1-\lambda_0\right)} \right| > \epsilon\right) = 0.
$$ \qed
\end{proof}

In Lemma \ref{lemma:Rn_gamma_to_Rn_zeta}, we show that the Simpson's Index at the cancer recurrence time is very close to that at the deterministic limit time of cancer recurrence.
\begin{lemma}
\label{lemma:Rn_gamma_to_Rn_zeta}
For any $\epsilon > 0$, when $u<(1-\alpha)/4$, we have     
\begin{align*}
    \lim_{n\to\infty}\PP_{\rho_n}\left(n^u \left|n^{1-\alpha} R_n\left(\gamma_n\right)-  n^{1-\alpha} R_n\left(\zeta_n\right)\right| > \epsilon\right) = 0.
\end{align*}
\end{lemma}
\begin{proof}
When $\zeta_n < \gamma_n$, we have
\begin{align}
&\PP_{\rho_n}\left(n^u\left|n^{1-\alpha}\left( R_n(\gamma_n)-R_n(\zeta_n)\right)\right|>\epsilon\right) \nonumber\\
=&\PP_{\rho_n}\left(n^u\left|n^{1-\alpha}\left( \sum_{i=1}^{\tilde{I}_n(\gamma_n)}\frac{X_{i,n}(\gamma_n)^2}{n^2}-\sum_{i=1}^{\tilde{I}_n(\zeta_n)}\frac{X_{i,n}(\zeta_n)^2}{Z_1^n(\zeta_n)^2} \right)\right|>\epsilon\right)\nonumber\\
\leq & \PP_{\rho_n}\left(n^u\left|n^{1-\alpha}\left( \sum_{i=1}^{\tilde{I}_n(\zeta_n)}\frac{X_{i,n}(\gamma_n)^2}{n^2}-\sum_{i=1}^{\tilde{I}_n(\zeta_n)}\frac{X_{i,n}(\zeta_n)^2}{Z_1^n(\zeta_n)^2} \right)\right|>\epsilon/2\right)\nonumber\\
& \quad +\PP_{\rho_n}\left( \mathlarger{\sum}_{i=\tilde{I}_n( \zeta_n)+1}^{\tilde{I}_n(\gamma_n)}\frac{X_{i,n}(\gamma_n)^2}{n^{1+\alpha-u}}>\epsilon/2\right)\nonumber\\
\leq & \PP_{\rho_n}\left(n^u\left|n^{1-\alpha}\left( \sum_{i=1}^{\tilde{I}_n(\zeta_n)}\frac{X_{i,n}(\gamma_n)^2}{n^2}-\sum_{i=1}^{\tilde{I}_n(\zeta_n)}\frac{X_{i,n}(\zeta_n)^2}{Z_1^n(\zeta_n)^2} \right)\right|>\epsilon/2\right) \label{eqn:first_term_recurrence_time_deterministic}\\
& \quad +\PP\left( \mathlarger{\sum}_{i=\tilde{I}_n( \zeta_n)+1}^{\tilde{I}_n(\gamma_n)}\frac{X_{i,n}(\gamma_n)^2}{n^{1+\alpha-u}}>\epsilon/2\right). \label{eqn:second_term_recurrence_time_deterministic}
\end{align}
Here we define $X_{i,n}(t) = 0$ if the creation time of $i_{th}$ clone is after $t$. Therefore, when $\zeta_n\geq \gamma_n$, we have 
\begin{align}
&\PP_{\rho_n}\left(n^u\left|n^{1-\alpha}\left( R_n(\gamma_n)-R_n(\zeta_n)\right)\right|>\epsilon\right) \nonumber\\
=&\PP_{\rho_n}\left(n^u\left|n^{1-\alpha}\left( \sum_{i=1}^{\tilde{I}_n(\gamma_n)}\frac{X_{i,n}(\gamma_n)^2}{n^2}-\sum_{i=1}^{\tilde{I}_n(\zeta_n)}\frac{X_{i,n}(\zeta_n)^2}{Z_1^n(\zeta_n)^2} \right)\right|>\epsilon\right)\nonumber\\
= & \PP_{\rho_n}\left(n^u\left|n^{1-\alpha}\left( \sum_{i=1}^{\tilde{I}_n(\zeta_n)}\frac{X_{i,n}(\gamma_n)^2}{n^2}-\sum_{i=\tilde{I}_n(\gamma_n)+1}^{\tilde{I}_n(\zeta_n)}\frac{X_{i,n}(\gamma_n)^2}{n^2} -\sum_{i=1}^{\tilde{I}_n(\zeta_n)}\frac{X_{i,n}(\zeta_n)^2}{Z_1^n(\zeta_n)^2} \right)\right|>\epsilon/2\right)
\label{eqn:second_term_recurrence_time_deterministic_split}
\end{align}
Because of our definition above, we have $X_{i,n}(\gamma_n) = 0$ when $\tilde{I}_n(\gamma_n)+1 \leq i \leq \tilde{I}_n(\zeta_n)$. Therefore, \ref{eqn:second_term_recurrence_time_deterministic_split} can be reduced to \ref{eqn:first_term_recurrence_time_deterministic}.

For \eqref{eqn:second_term_recurrence_time_deterministic}, we can obtain that
\begin{align*}
    & \quad \PP\left( \sum_{i=\tilde{I}_n( \zeta_n)+1}^{\tilde{I}_n(\gamma_n)}\frac{X_{i,n}(\gamma_n)^2}{n^{1+\alpha-u}}>\epsilon/2\right)\\
    & \leq \PP\left(\sum_{i=\tilde{I}_n( \zeta_n)+1}^{\tilde{I}_n(\gamma_n)}\frac{X_{i,n}(\gamma_n)^2}{n^{1+\alpha-u}}>\epsilon/2,\gamma_n\leq \zeta_n+\delta_n\right)+ \PP\left( \gamma_n>\zeta_n+\delta_n\right),
\end{align*}
where $\delta_n$ is some constant, and the summation is zero if $\zeta_n > \gamma_n$. From Theorem \ref{cip_gamma_faster}, we know that when $\delta_n = n^{-v}, v < (1-\alpha)/2 $, $ \PP\left( \gamma_n>\zeta_n+\delta_n\right)$ converges to 0. Moreover, we have
\begin{align*}
&\PP\left( \sum_{i=\tilde{I}_n( \zeta_n)+1}^{\tilde{I}_n(\gamma_n)}\frac{X_{i,n}(\gamma_n)^2}{n^{1+\alpha-u}}>\epsilon/2,\gamma_n\leq \zeta_n+\delta_n\right)\\
\leq & \PP\left( \sum_{i=\tilde{I}_n( \zeta_n)+1}^{\tilde{I}_n(\zeta_n+\delta_n)}\frac{X_{i,n}(\gamma_n)^2}{n^{1+\alpha-u}}>\epsilon/2,\gamma_n\leq \zeta_n+\delta_n\right)\\
\leq & \PP\left( \sum_{i=\tilde{I}_n( \zeta_n)+1}^{\tilde{I}_n(\zeta_n+\delta_n)}\frac{Z^{(i)}(\delta_n)^2}{n^{1+\alpha-u}}>\epsilon/2\right),
\end{align*}
where $Z^{(i)}(t)$'s are i.i.d copies of $Z(t)$, which is the population size of a death-birth process starting from a single resistant cell. By Markov's inequality and 
\begin{align*}
\EE\left[ \tilde{I}_n(\zeta_n+\delta_n)-\tilde{I}_n( \zeta_n)\right] &= n^{1-\alpha}\int_{\zeta_n}^{\zeta_n+\delta_n}e^{\lambda_0 s}ds \\
& \sim n^{1-\alpha} e^{\lambda_0\zeta_n} \delta_n \sim n^{1-\alpha+\lambda_0\alpha/\lambda_1-v},
\end{align*}
we have 
\begin{align*}
    \PP\left( \sum_{i=\tilde{I}_n( \zeta_n)+1}^{\tilde{I}_n(\zeta_n+\delta_n)}\frac{Z^{(i)}(\delta_n)^2}{n^{1+\alpha-u}}>\epsilon/2\right) &\leq \frac{2}{\epsilon n^{1+\alpha-u}} \EE\left[ \tilde{I}_n(\zeta_n+\delta_n)-\tilde{I}_n( \zeta_n)\right]\EE\left[ Z(\delta_n)^2\right]\\
    &= \Theta(n^{u - 2\alpha-v +\lambda_0\alpha/\lambda_1}) \to 0.
\end{align*}
Next we analyze \eqref{eqn:first_term_recurrence_time_deterministic}. We have
\begin{align}
&\PP_{\rho_n}\left(n^u\left|n^{1-\alpha}\left( \sum_{i=1}^{\tilde{I}_n(\zeta_n)}\frac{X_{i,n}(\gamma_n)^2}{n^2}-\sum_{i=1}^{\tilde{I}_n(\zeta_n)}\frac{X_{i,n}(\zeta_n)^2}{Z_1^n(\zeta_n)^2} \right)\right|>\epsilon/2\right) \nonumber\\
\leq&\PP\left(n^u\left|n^{1-\alpha}\left( \sum_{i=1}^{\tilde{I}_n(\zeta_n)}\frac{X_{i,n}(\gamma_n)^2}{n^2}-\sum_{i=1}^{\tilde{I}_n(\zeta_n)}\frac{X_{i,n}(\zeta_n)^2}{n^2} \right)\right|>\epsilon/2\right) \label{eqn:same_n_term}\\
&+  \PP_{\rho_n}\left(n^u\left|n^{1-\alpha}\left( \sum_{i=1}^{\tilde{I}_n(\zeta_n)}\frac{X_{i,n}(\zeta_n)^2}{n^2}-\sum_{i=1}^{\tilde{I}_n(\zeta_n)}\frac{X_{i,n}(\zeta_n)^2}{Z_1^n(\zeta_n)^2} \right)\right|>\epsilon/2\right). \label{eqn:same_time_term}
\end{align}
We work on \eqref{eqn:same_time_term} first and aim to show that
\begin{align}
    \lim\limits_{n\to\infty}\PP_{\rho_n}\left(n^u\left|n^{1-\alpha}\left( \sum_{i=1}^{\tilde{I}_n(\zeta_n)}\frac{X_{i,n}(\zeta_n)^2}{n^2}-\sum_{i=1}^{\tilde{I}_n(\zeta_n)}\frac{X_{i,n}(\zeta_n)^2}{Z_1^n(\zeta_n)^2} \right)\right|>\epsilon\right)=0. \label{eqn:desired_result_same_time_term}
\end{align}
For the term within the probability, we have
\begin{align*}
    &n^{1-\alpha+u}\left| \sum_{i=1}^{\tilde{I}_n(\zeta_n)}\frac{X_{i,n}(\zeta_n)^2}{n^2}-\sum_{i=1}^{\tilde{I}_n(\zeta_n)}\frac{X_{i,n}(\zeta_n)^2}{Z_1^n(\zeta_n)^2} \right|\\
    =& n^{1-\alpha+u}\left| \frac{1}{n^2}-\frac{1}{Z_1^n(\zeta_n)^2} \right|\sum_{i=1}^{\tilde{I}_n(\zeta_n)}X_{i,n}(\zeta_n)^2\\
    =& n^u\left( 1- \frac{n^2}{Z_1^n(\zeta_n)^2} \right)n^{1-\alpha}\sum\limits_{i=1}^{\tilde{I}_n(\zeta_n)}n^{-2}X_{i,n}(\zeta_n)^2\\
    =& n^u\left( 1- \frac{n^2}{Z_1^n(\zeta_n)^2} \right)n^{1-\alpha}\tilde{R}_n(\zeta_n)
\end{align*}
In the proof of Lemma \ref{lemma:Convergence_of_Rn_zeta_n}, we have shown that $n^{1-\alpha}\tilde{R}_n(\zeta_n)$ converges to a constant in probability. Moreover, by a direct calculation of $\EE\left[Z_1^n(\zeta_n)^4\right]$ (see the proof of Proposition 1 in \cite{CD2021}), we have $\va\left(n^u\left( 1- \frac{Z_1^n(\zeta_n)^2}{n^2} \right)\right)= \Theta(n^{4\alpha-4+2u})$.
Because $u<(1-\alpha)/4$, we have 
\begin{align*}
    \lim_{n\to \infty} \PP\left(n^u\left| 1- \frac{Z_1^n(\zeta_n)^2}{n^2} \right|>\epsilon\right) = 0.
\end{align*}
By some algebra, we have when $n^u\left| 1- \frac{Z_1^n(\zeta_n)^2}{n^2} \right|<\epsilon$, 
\begin{align*}
    \frac{-\epsilon}{1-\epsilon n^{-u}} < n^u\left( 1- \frac{n^2}{Z_1^n(\zeta_n)^2} \right) <  \frac{\epsilon}{1+\epsilon n^{-u}},
\end{align*}
which implies
\begin{align*}
    \lim_{n\to \infty} \PP_{\rho_n}\left(n^u\left| 1- \frac{n^2}{Z_1^n(\zeta_n)^2} \right|>\epsilon\right) = 0.
\end{align*}
Therefore, by the continuous mapping theorem, $n^u\left( 1- \frac{n^2}{Z_1^n(\zeta_n)^2} \right)n^{1-\alpha}\tilde{R}_n(\zeta_n)$ converges to $0$ in probability and we have shown the desired result in \eqref{eqn:desired_result_same_time_term}.

Next, we work on \eqref{eqn:same_n_term} and aim to show that
\begin{align*}
   \lim\limits_{n\to \infty} \PP\left(n^u\left|n^{1-\alpha}\left( \sum_{i=1}^{\tilde{I}_n(\zeta_n)}\frac{X_{i,n}(\zeta_n)^2-X_{i,n}(\gamma_n)^2}{n^2} \right)\right|>\epsilon\right) = 0.
\end{align*}
For simplicity, we define the event $A_n = \left\{ \omega\big| |\gamma_n-\zeta_n| < \delta_n \right\}$ where $\delta_n = \delta n^{-v}$ and $\delta$ is a fixed constant. By Theorem \ref{cip_gamma_faster}, we know that when $v<(1-\alpha)/2$, $\lim\limits_{n\to\infty}\PP\left(A_n^C\right) = 0$. Hence, we have
\begin{align}
&\quad \PP\left(n^u\left|n^{1-\alpha}\left( \sum_{i=1}^{\tilde{I}_n(\zeta_n)}\frac{X_{i,n}(\zeta_n)^2-X_{i,n}(\gamma_n)^2}{n^2} \right)\right|>\epsilon\right) \nonumber\\
&\leq \PP\left(n^u\left|n^{1-\alpha}\left( \sum_{i=1}^{\tilde{I}_n(\zeta_n)}\frac{X_{i,n}(\zeta_n)^2-X_{i,n}(\gamma_n)^2}{n^2} \right)\right|>\epsilon, A_n\right)
+ \PP\left(A_n^C\right). \label{eqn:two_terms_for_eqn_same_n_term}
\end{align}
For the first term in \eqref{eqn:two_terms_for_eqn_same_n_term}, we have 
\begin{align*}
    & \quad \left|n^{1-\alpha+u}\left( \sum_{i=1}^{\tilde{I}_n(\zeta_n)}\frac{X_{i,n}(\zeta_n)^2-X_{i,n}(\gamma_n)^2}{n^2} \right)\right|\cdot 1_{A_n}\\
    &\leq n^{-1+\alpha+u}\sup\limits_{0\leq\Delta t\leq 2\delta_n}n^{-2\alpha}\left|\sum_{i=1}^{\tilde{I}_n(\zeta_n)} \left(X_{i,n}(\zeta_n)^2 -X_{i,n}(\zeta_n-\delta_n+\Delta t)^2\right)\right|\\
    &\leq n^{-1+\alpha+u}\sum_{i=1}^{\tilde{I}_n(\zeta_n)}\sup\limits_{0\leq\Delta t\leq 2\delta_n}n^{-2\alpha}\left| X_{i,n}(\zeta_n)^2 -X_{i,n}(\zeta_n-\delta_n+\Delta t)^2\right|\\
    &\leq n^{-1+\alpha+u}\sum_{i=1}^{\tilde{I}_n(\zeta_n)} n^{-2\alpha}\left|  X_{i,n}(\zeta_n)^2 -X_{i,n}(\zeta_n-\delta_n )^2\right|\\
    & \quad +n^{-1+\alpha+u}\sum_{i=1}^{\tilde{I}_n(\zeta_n)}\sup\limits_{0\leq\Delta t\leq 2\delta_n}n^{-2\alpha}\left| X_{i,n}(\zeta_n-\delta_n +\Delta t )^2 -X_{i,n}(\zeta_n-\delta_n )^2\right|.
\end{align*}
Therefore, by Markov's inequality,
\begin{align}
    &\PP\left(n^u\left|n^{1-\alpha}\left( \sum_{i=1}^{\tilde{I}_n(\zeta_n)}\frac{X_{i,n}(\zeta_n)^2-X_{i,n}(\gamma_n)^2}{n^2} \right)\right|>\epsilon, A_n\right) \nonumber\\
    \leq & \frac{1}{\epsilon}\EE\left[ n^{-1+\alpha+u}\sum_{i=1}^{\tilde{I}_n(\zeta_n)}\sup\limits_{0\leq\Delta t\leq 2\delta_n}n^{-2\alpha}\left| X_{i,n}(\zeta_n-\delta_n+\Delta t)^2 -X_{i,n}(\zeta_n-\delta_n)^2\right|\right] \label{eqn:markov_first_term}\\
    + & \frac{1}{\epsilon}\EE\left[ n^{-1+\alpha+u}\sum_{i=1}^{\tilde{I}_n(\zeta_n)}n^{-2\alpha}\left| X_{i,n}(\zeta_n)^2 -X_{i,n}(\zeta_n-\delta_n)^2\right|\right].\label{eqn:markov_second_term}
\end{align}
We bound \eqref{eqn:markov_first_term} first. Note that
\begin{align*}
    n^{-1+\alpha}\EE\left[ \tilde{I}_n(\zeta_n) \right]\to C,
\end{align*}
where $C$ is some constant. If we condition on the value of $\tilde{I}_n(\zeta_n)$, then by the uniformity property of Poisson process, the generation time of each clone is independently and identically distributed, and we denote by $\tau$ a random variable which follows such distribution. In addition, because $\tilde{I}_n(\zeta_n)$ and  $X_{i,n}$ are independent, by Wald's equation, we have 

\begin{align*}
    & \quad \frac{1}{\epsilon}\EE\left[ n^{-1+\alpha+u}\sum_{i=1}^{\tilde{I}_n(\zeta_n)}\sup\limits_{0\leq\Delta t\leq 2\delta_n}n^{-2\alpha}\left| X_{i,n}(\zeta_n-\delta_n+\Delta t)^2 -X_{i,n}(\zeta_n-\delta_n)^2\right|\right]\\
    &\sim \frac{C}{\epsilon}n^{-2\alpha+u}\EE\left[ \sup\limits_{0\leq\Delta t\leq 2\delta_n}\left| Z(\zeta_n-\delta_n-\tau+\Delta t)^2 -Z(\zeta_n-\delta_n-\tau)^2\right|\right]\\
    &\leq \frac{C}{\epsilon}n^{-2\alpha+u}\EE\left[ \left(\sup\limits_{0\leq\Delta t\leq 2\delta_n}\left| Z(\zeta_n-\delta_n-\tau+\Delta t)^2 -Z(\zeta_n-\delta_n-\tau)^2\right|\right)^2\right]^{1/2}\\
    &= \frac{C}{\epsilon}n^{-2\alpha+u}\EE\left[ \sup\limits_{0\leq\Delta t\leq 2\delta_n}\left( Z(\zeta_n-\delta_n-\tau+\Delta t)^2 -Z(\zeta_n-\delta_n-\tau)^2\right)^2\right]^{1/2}\\
    &= \frac{C}{\epsilon}n^{-2\alpha+u}\EE\left[\EE\left[ \sup\limits_{0\leq\Delta t\leq 2\delta_n}\left( Z(\zeta_n-\delta_n-\tau+\Delta t)^2 -Z(\zeta_n-\delta_n-\tau)^2\right)^2\Big|\tau\right]\right]^{1/2},
\end{align*}
where $Z(t)$ is the population size of a branching process starting from a single resistant cell. It's noteworthy that $t$ can be less than 0 for $Z(t)$. As a natural extension from the definition of $X_{i,n}(t)$, we set $Z(t) = 0$ when $t<0$. By Doob's Maximal Inequality and the independence between $\tau$ and $Z(t)$, we have
\begin{align*}
    & \quad \frac{C}{\epsilon}n^{-2\alpha+u}\EE\left[\EE\left[ \sup\limits_{0\leq\Delta t\leq 2\delta_n}\left( Z(\zeta_n-\delta_n-\tau+\Delta t)^2 -Z(\zeta_n-\delta_n-\tau)^2\right)^2\Big|\tau\right]\right]^{1/2}\\
    & \leq \frac{4C}{\epsilon}n^{-2\alpha+u}\EE\left[\EE\left[ \left( Z(\zeta_n-\tau+\delta_n)^2 -Z(\zeta_n-\tau-\delta_n)^2\right)^2\Big|\tau\right]\right]^{1/2}\\
    & = \frac{4C}{\epsilon}n^{-2\alpha+u}\EE\left[\EE\left[  Z(\zeta_n-\tau+\delta_n)^4 +Z(\zeta_n-\tau-\delta_n)^4- 2 Z(\zeta_n-\tau-\delta_n)^2Z(\zeta_n-\tau+\delta_n)^2\Big|\tau\right]\right]^{1/2}\\
    & \overset{\text{(a)}}{\leq} \frac{4C}{\epsilon}n^{-2\alpha+u}\EE\left[\EE\left[  Z(\zeta_n-\tau+\delta_n)^4 -Z(\zeta_n-\tau-\delta_n)^4\Big|\tau\right]\right]^{1/2}\\
    & \overset{\text{(b)}}{\leq} \frac{4C}{\epsilon}n^{-2\alpha+u}\EE\left[  Z(\zeta_n+\delta_n)^4 -Z(\zeta_n-\delta_n)^4\right]^{1/2} \\
    & \overset{\text{(c)}}{=} \Theta(n^{u-v/2}),
\end{align*}
where we use the fact that $Z(t)^2$ is a submartingale in inequality (a); $\EE\left[Z(t)^4\right]$ is increasing in $t$ in inequality (b); $\EE\left[Z(t)^4\right] = \Theta\left(e^{4\lambda_1t}\right)$ and $e^{4\lambda_1\zeta_n}=\Theta(n^{4\alpha})$ in equality (c).

Finally, we bound \eqref{eqn:markov_second_term}. By a similar argument, we have
\begin{align*}
    & \quad \frac{1}{\epsilon}\EE\left[ n^{-1+\alpha+u}\sum_{i=1}^{\tilde{I}_n(\zeta_n)}n^{-2\alpha}\left| X_{i,n}(\zeta_n)^2 -X_{i,n}(\zeta_n-\delta_n)^2\right|\right]\\
    &\sim \frac{C}{\epsilon}n^{-2\alpha+u}\EE\left[  \left| Z(\zeta_n-\tau)^2 -Z(\zeta_n-\tau-\delta_n)^2\right|\right]\\
    &\leq \frac{C}{\epsilon}n^{-2\alpha+u}\EE\left[  \left( Z(\zeta_n-\tau)^2 -Z(\zeta_n-\tau-\delta_n)^2\right)^2\right]^{1/2}\\
    & \leq \frac{C}{\epsilon}n^{-2\alpha+u}\EE\left[  Z(\zeta_n-\tau)^4 +Z(\zeta_n-\tau-\delta_n)^4- 2 Z(\zeta_n-\tau-\delta_n)^2Z(\zeta_n-\tau)^2\right]^{1/2}\\
    & \leq \frac{C}{\epsilon}n^{-2\alpha+u}\EE\left[  Z(\zeta_n-\tau)^4 -Z(\zeta_n-\tau-\delta_n)^4\right]^{1/2}\\
    & \leq \frac{C}{\epsilon}n^{-2\alpha+u}\EE\left[  Z(\zeta_n)^4 -Z(\zeta_n-\delta_n)^4\right]^{1/2} \\
    & = \Theta(n^{u-v/2}).
\end{align*}
Because $u<(1-\alpha)/4$, we can find $v<(1-\alpha)/2$ so that $u<v/2$ which, combined with the previous analysis, implies that
\begin{align*}
   \lim\limits_{n\to \infty} \PP\left(n^u\left|n^{1-\alpha}\left( \sum_{i=1}^{\tilde{I}_n(\zeta_n)}\frac{X_{i,n}(\zeta_n)^2-X_{i,n}(\gamma_n)^2}{n^2} \right)\right|>\epsilon\right) = 0.
\end{align*}
To conclude, when $u<(1-\alpha)/4 $ we have     
\begin{align*}
    \lim_{n\to\infty}\PP_{\rho_n}\left(n^u \left|n^{1-\alpha} R_n\left(\gamma_n\right)-  n^{1-\alpha} R_n\left(\zeta_n\right)\right| > \epsilon\right) = 0.
\end{align*} \qed
\subsubsection*{Proof of Theorem \ref{cip_Rn}:}
    Theorem \ref{cip_Rn} follows directly from Proposition \ref{prop:rho_n}, Lemma \ref{lemma:Convergence_of_Rn_zeta_n} and Lemma \ref{lemma:Rn_gamma_to_Rn_zeta}.
\end{proof}
\subsection{Proof of Theorem \ref{thm: consistent estimators}}
\label{sec: consistent estimators}
\begin{lemma}
    \label{lemma: Z0_gamma>0}
    When $\alpha<-\frac{\lambda_1}{\lambda_0}$
    \begin{align*}
\lim_{n\to\infty}\PP\left(Z_0^n(\gamma_n)=0  \right) = 0
    \end{align*}
\end{lemma}
\begin{proof}
    \begin{align*}
        \PP\left(Z_0^n(\gamma_n)=0  \right) & \leq \PP\left(Z_0^n(\gamma_n)=0 , \gamma_n\leq\zeta_n+\epsilon \right) +\PP\left(\gamma_n>\zeta_n+\epsilon  \right) \\
&\leq \PP\left(Z_0^n(\zeta_n +\epsilon)=0  \right) +\PP\left(\gamma_n>\zeta_n+\epsilon  \right)
    \end{align*}
    From Theorem \ref{cip_gamma_faster} we know $\PP\left(\gamma_n>\zeta_n+\epsilon  \right)\to 0$  and we have
    \begin{align*}
        \PP\left(Z_0^n(\zeta_n +\epsilon)=0  \right) &\leq \PP\left(Z_0^n(\zeta_n +\epsilon)<ne^{\lambda_0(\zeta_n+\epsilon)}/2 \right)\\
        &\leq \PP\left(\left|Z_0^n(\zeta_n +\epsilon)-ne^{\lambda_0(\zeta_n+\epsilon)}\right|>ne^{\lambda_0(\zeta_n+\epsilon)}/2 \right)\\
        &\leq 4\var(Z_0^n(\zeta_n+\epsilon))/n^2e^{2\lambda_0(\zeta_n+\epsilon)} 
    \end{align*}
    From \ref{Order: second moment Z^n}, we know $$4\var(Z_0^n(\zeta_n+\epsilon))/n^2e^{2\lambda_0(\zeta_n+\epsilon)} = \frac{4n(r_0+d_0)e^{\lambda_0(\zeta_n+\epsilon)}(e^{\lambda_0(\zeta_n+\epsilon)}-1)}{\lambda_0n^2e^{2\lambda_0(\zeta_n+\epsilon)}} \sim \frac{1}{ne^{\lambda_0(\zeta_n+\epsilon)}}\to 0$$
    \qed
\end{proof}
\subsubsection*{Proof of Theorem \ref{thm: consistent estimators}}
\begin{proof}
    \begin{itemize}
        \item \boldsymbol{$\hat{\lambda}^{(n)}_0$}: We show a stronger convergence result for $\hat{\lambda}^{(n)}_0$. That is for $\epsilon > 0$ and $u<\min\left\{  (1-\alpha)/2, (1+\lambda_0\alpha/\lambda_1)/2\right\}$,
        \begin{align*}
           \lim_{n\to \infty}\PP\left( n^{u}\left| \hat{\lambda}^{(n)}_0-\lambda_0  \right| > \epsilon\right) =0.
        \end{align*}
Let $\epsilon_n = \epsilon n^{-v} $, where $u<v<(1-\alpha)/2$ By Theorem \ref{cip_gamma_faster} and Lemma \ref{lemma: Z0_gamma>0} we have
\begin{align*}
& \lim_{n\to\infty}\PP\left(\left|\zeta_n-\gamma_n\right|>\epsilon_n\right)= 0, \text{and}\\
& \lim_{n\to\infty}\PP\left(Z_0^n(\gamma_n)=0  \right) = 0.
\end{align*}
        Define $A_n = \left\{ \omega | \left|\zeta_n-\gamma_n\right|<\epsilon_n, 
 Z_0^n(\gamma_n)>0  \right\}$ and we have $\lim\limits_{n\to\infty}\PP\left(A_n^C  \right) = 0$. When $\omega \in A_n$, we have $\zeta_n-\epsilon_n<\gamma_n$ and $ 
 Z_0^n(\gamma_n)>0$, and therefore we must have $Z_0^n(\zeta_n-\epsilon_n)>0$ and
        \begin{align*}
           & n^{u}\left| \hat{\lambda}^{(n)}_0-\lambda_0  \right| \\
           = & n^{u} \left| \frac{1}{\gamma_n}\log{\frac{Z^n_0(\gamma_n)}{n}}-\lambda_0  \right|\\
             = &n^{u} \bigg|\frac{1}{\gamma_n}\log{\frac{Z^n_0(\gamma_n)}{n}}-\frac{1}{\gamma_n}\log{\frac{Z^n_0(\zeta_n-\epsilon_n)}{n}}+\frac{1}{\gamma_n}\log{\frac{Z^n_0(\zeta_n-\epsilon_n)}{n}}-\frac{1}{\zeta_n}\log{\frac{Z^n_0(\zeta_n-\epsilon_n)}{n}}  \\
             &+\frac{1}{\zeta_n}\log{\frac{Z^n_0(\zeta_n-\epsilon_n)}{n}}-\lambda_0\bigg|\\
             \leq & \frac{n^{u}}{\gamma_n}\left| \log \frac{Z^n_0(\gamma_n)}{Z^n_0(\zeta_n-\epsilon_n)} \right|+\frac{n^u\epsilon_n}{\gamma_n \zeta_n}  \log{\frac{n}{Z^n_0(\zeta_n-\epsilon_n)}} + \frac{n^u}{\zeta_n}\left|\log{\frac{Z^n_0(\zeta_n-\epsilon_n)}{ne^{\lambda_0\zeta_n}}}\right|.
        \end{align*}
        Because $\zeta_n = \Theta(\log n)$, $\zeta_n-\gamma_n$ converges to $0$ in probability, and $\log{\frac{n}{Z^n_0(\zeta_n-\epsilon_n)}}< \log{n} $, the second term converges to $0$ in probability. It remains to show that $\frac{n^{u}}{\log{n}}\left| \log \frac{Z^n_0(\gamma_n)}{Z^n_0(\zeta_n-\epsilon_n)} \right| $  and $\frac{n^{u}}{\log{n}}\left| \log \frac{Z^n_0(\zeta_n-\epsilon_n)}{ne^{\lambda_0\zeta_n}} \right| $ converge to $0$ in probability, which are equivalent to
        \begin{align}
            & \left(\left(\frac{Z^n_0(\gamma_n)}{Z^n_0(\zeta_n-\epsilon_n)}\vee \frac{Z^n_0(\zeta_n-\epsilon_n)}{Z^n_0(\gamma_n)}\right)-1 +1\right)^{n^u/\log{n}}\xrightarrow{p} 1, \label{converge_1_twoterm_1}\\ 
            & \left(\left(\frac{Z^n_0(\zeta_n-\epsilon_n)}{ne^{\lambda_0\zeta_n}}\vee \frac{ne^{\lambda_0\zeta_n}}{Z^n_0(\zeta_n-\epsilon_n)}\right)-1 +1\right)^{n^u/\log{n}} \xrightarrow{p} 1. \label{converge_1_twoterm_2}
        \end{align}
        

        Consider a non-negative random variable $X(n)$ such that there exists a positive constant $C$ such that $\lim\limits_{n\rightarrow \infty} \PP( n^uX(n) \leq  C)=1$. Then on the event  $\{n^uX(n)\leq C\}$, 
        we have
        \begin{align*}
            &\left(X(n)+1\right)^{n^u/\log{n}} \\
            =& \left(\frac{n^uX(n)}{\log n}\frac{\log n}{n^u}+1\right)^{n^u/\log{n}}\\
            \leq& \exp\left[\frac{C}{\log n}\right],
        \end{align*}
        which implies that $\left(X(n)+1\right)^{n^u/\log{n}}$ converges to $1$ in probability. With this result, to prove \eqref{converge_1_twoterm_1} and \eqref{converge_1_twoterm_2}, it suffices to show that
        \begin{align}
        n^u\left(\left(\frac{Z^n_0(\gamma_n)}{Z^n_0(\zeta_n-\epsilon_n)}\vee \frac{Z^n_0(\zeta_n-\epsilon_n)}{Z^n_0(\gamma_n)}\right)-1\right)  \label{eqn:estimator_lambda_0_bound_1}
        \end{align}
        and 
        \begin{align}
        n^u\left(\left(\frac{Z^n_0(\zeta_n-\epsilon_n)}{ne^{\lambda_0\zeta_n}}\vee \frac{ne^{\lambda_0\zeta_n}}{Z^n_0(\zeta_n-\epsilon_n)}\right)-1\right) \label{eqn:estimator_lambda_0_bound_2}
        \end{align}
        are bounded by a constant in probability. We analyze \eqref{eqn:estimator_lambda_0_bound_2} first. From \eqref{Order: second moment Z^n} and the condition that $u<1/2 +\lambda_0 \alpha/2\lambda_1$, we have 
        \begin{align*}
            &\var\left(n^u\frac{Z^n_0(\zeta_n-\epsilon_n)}{ne^{\lambda_0\zeta_n}}\right)
            =\Theta \left(n^{2u-1}e^{-\lambda_0 (\zeta_n+\epsilon_n)}\right) = \Theta(n^{-\frac{\lambda_0 \alpha}{\lambda_1}-1+2u})\to 0, \text{ and} \\
            &\EE\left[ n^u\left(\frac{Z^n_0(\zeta_n-\epsilon_n)}{ne^{\lambda_0\zeta_n}}-1\right)\right] =n^u\left( e^{-\lambda_0 \epsilon_n}-1\right) \to 0.
        \end{align*}
        Therefore, $n^u\left(Z^n_0(\zeta_n-\epsilon_n)/ne^{\lambda_0\zeta_n}-1\right)$ converges to $0$ in probability. By a similar argument to that in the proof of Lemma \ref{lemma:Rn_gamma_to_Rn_zeta}, we can obtain that \eqref{eqn:estimator_lambda_0_bound_2} is bounded by a constant in probability.
        
        We then analyze \eqref{eqn:estimator_lambda_0_bound_1}. 
        \begin{align*}
            &n^u\left(\frac{Z^n_0(\gamma_n)}{Z^n_0(\zeta_n-\epsilon_n)} -1\right) \\
            =& n^u\left(\frac{Z^n_0(\gamma_n)}{Z^n_0(\zeta_n-\epsilon_n)} -1\right)1_{\{\gamma_n\in \left(\zeta_n-\epsilon_n, \zeta_n+\epsilon_n\right)\}}+n^u\left(\frac{Z^n_0(\gamma_n)}{Z^n_0(\zeta_n-\epsilon_n)} -1\right)1_{\{\gamma_n\notin \left(\zeta_n-\epsilon_n, \zeta_n+\epsilon_n\right)\}}.
        \end{align*}
        By Theorem \ref{cip_gamma_faster}, we can safely discard the second term. For the first term, we have
        \begin{align*}
            & \quad n^u\left(\frac{Z^n_0(\gamma_n)}{Z^n_0(\zeta_n-\epsilon_n)} -1\right)1_{\{\gamma_n\in \left(\zeta_n-\epsilon_n, \zeta_n+\epsilon_n\right)\}}\\
            & \le 1_{\{\gamma_n\in \left(\zeta_n-\epsilon_n, \zeta_n+\epsilon_n\right)\}} \frac{1}{Z^n_0(\zeta_n-\epsilon_n)}\sum_{i = 1}^{Z^n_0(\zeta_n - \epsilon_n)}n^u\left(\sup\limits_{t\in [0, 2\epsilon_n]}B_0^{(i)}(t)-1\right), \text{ and}
        \end{align*}
        \begin{align*}
            & \quad n^u\left(\frac{Z^n_0(\gamma_n)}{Z^n_0(\zeta_n-\epsilon_n)} -1\right)1_{\{\gamma_n\in \left(\zeta_n-\epsilon_n, \zeta_n+\epsilon_n\right)\}}\\
            & \ge 1_{\{\gamma_n\in \left(\zeta_n-\epsilon_n, \zeta_n+\epsilon_n\right)\}} \frac{1}{Z^n_0(\zeta_n-\epsilon_n)}\sum_{i = 1}^{Z^n_0(\zeta_n - \epsilon_n)}n^u\left(\inf\limits_{t\in [0, 2\epsilon_n]}B_0^{(i)}(t)-1\right),
        \end{align*}
        where $\{B_0^{(i)}(t)\}$'s are i.i.d copies of the size of a branching process starting from a single sensitive cell; the upper bound is non-negative and the lower bound is non-positive. Notice that the right-continuous process $\{B_0^i(t)\}$'s are  independent with $Z^n_0(\zeta_n-\epsilon_n)$. Hence, we have 
        \begin{align*}
           &\lim_{n\to\infty} \EE\left[\frac{1}{Z^n_0(\zeta_n-\epsilon_n)}\sum_{i = 1}^{Z^n_0(\zeta_n - \epsilon_n)}n^u\left(\sup\limits_{t\in [0, 2\epsilon_n]}B_0^{(i)}(t)-1\right)   \right] \\
           =& \lim_{n\to\infty} \EE\left[ n^u\left(\sup\limits_{t\in [0, 2\epsilon_n]}B_0^{(i)}(t)-1\right)  \right] \\
           \le &  \lim_{n\to\infty}  n^u\left(\EE\left[\sup\limits_{t\in [0, 2\epsilon_n]}\hat{B}_0^{(i)}(t)\right] -1\right)\\
           =& \lim_{n\to\infty}  n^u\left(\EE\left[\hat{B}_0^{(i)}(2\epsilon_n)\right] -1\right)\\
        =&  \lim_{n\to\infty}  n^u\left(e^{2r_0\epsilon_n} -1\right)\\
        =& 0,
        \end{align*}
       where $\hat{B}_0^{(i)}(t)$'s are i.i.d copies of the size of branching process starting from a cell with birth rate $r_0$ and death rate $0$. Meanwhile,
        \begin{align*}
           &\lim_{n\to\infty} \EE\left[\frac{1}{Z^n_0(\zeta_n-\epsilon_n)}\sum_{i = 1}^{Z^n_0(\zeta_n - \epsilon_n)}n^u\left(1-\inf\limits_{t\in [0, 2\epsilon_n]}B_0^{(i)}(t)\right)   \right] \\
           =& \lim_{n\to\infty} \EE\left[ n^u\left(1- \inf\limits_{t\in [0, 2\epsilon_n]}B_0^{(i)}(t)\right)  \right] \\
           \geq&  \lim_{n\to\infty}  n^u\left(1-\EE\left[\inf\limits_{t\in [0, 2\epsilon_n]}\tilde{B}_0^{(i)}(t)\right] \right)\\
           =& \lim_{n\to\infty}  n^u\left(1-\EE\left[\tilde{B}_0^{(i)}(2\epsilon_n)\right] \right)\\
            =&  \lim_{n\to\infty}  n^u\left(1-e^{-2d_0\epsilon_n} \right)\\
            =& 0,
        \end{align*}
        where $\tilde{B}_0^{(i)}(t)$'s are i.i.d copies of the size of branching process starting from a cell with birth rate $0$ and death rate $d_0$. Therefore, by Markov's inequality, the upper bound and lower bound of $n^u\left(Z_0^n(\gamma_n)/Z_0^n(\zeta_n-\epsilon_n)-1\right)$ both converge to $0$ in probability, which implies that $n^u\left(Z_0^n(\gamma_n)/Z_0^n(\zeta_n-\epsilon_n)-1\right)$ converges to $0$ in probability as well. Combined with the previous analysis, we have shown that 
    \begin{align*}
        \lim_{n\to\infty}P\left(n^u \left|\hat{\lambda}^{(n)}_0 - \lambda_0 \right|>\epsilon\right) = 0.
    \end{align*}
    \item \boldsymbol{$\hat{\lambda}^{(n)}_1$}:
    Similar to $\hat{\lambda}^{(n)}_0$, we can prove a stronger convergence result for $\hat{\lambda}^{(n)}_1$. Specifically, we can show that for any $u < \min\{(1-\alpha)/4,  -\lambda_0\alpha/2\lambda_1, (1+\lambda_0\alpha/\lambda_1)/2 \}$ and $\epsilon > 0$,
        \begin{align*}
           \lim_{n\to \infty}\PP\left( n^{u}\left| \hat{\lambda}^{(n)}_1-\lambda_1  \right| > \epsilon\right) =0.
        \end{align*}
    Recall that $U_n =  \frac{\sqrt{I_n\cdot R_{n}}}{\sqrt{I_n\cdot R_{n}-2}} - 1$. By Theorem \ref{cip_In} and Theorem \ref{cip_Rn}, we can show that for any $v<\min\{(1-\alpha)/4,  -\lambda_0\alpha/2\lambda_1 \}$,
        \begin{align}
           \lim_{n\to \infty}\PP\left( n^{v}\left| U_n-U  \right| > \epsilon\right) =0,\label{eqn:u_n_and_u_in_probability}
        \end{align}
    where $U=-\frac{\lambda_0}{\lambda_1}$. Note that
     \begin{align*}
       \hat{\lambda}^{(n)}_1 - \lambda_1  &= -\frac{\hat{\lambda}^{(n)}_0}{U_n} + \frac{\hat{\lambda}^{(n)}_0}{U} - \frac{\hat{\lambda}^{(n)}_0}{U} + \frac{\lambda_0}{U}\\
       & = -\hat{\lambda}^{(n)}_0\left( \frac{1}{U_n}-\frac{1}{U} \right)  - \frac{1}{U}\left( \hat{\lambda}^{(n)}_0 - \lambda_0\right).
    \end{align*}
    By \eqref{eqn:u_n_and_u_in_probability} and the result on $\hat{\lambda}^{(n)}_0$ such that for any $w<\min\left\{  (1-\alpha)/2, (1+\lambda_0\alpha/\lambda_1)/2\right\}$,
        \begin{align*}
        \lim_{n\to\infty}P\left(n^w \left|\hat{\lambda}^{(n)}_0 - \lambda_0 \right|>\epsilon\right) = 0,
    \end{align*}
    we can obtain that for any $u < \min\{(1-\alpha)/4,  -\lambda_0\alpha/2\lambda_1 ,(1+\lambda_0\alpha/\lambda_1)/2\} $, we have
    \begin{align}
          \lim_{n\to\infty}P\left(n^u \left|\hat{\lambda}^{(n)}_1 - \lambda_1 \right|>\epsilon\right) = 0. \label{eqn:lambda_1_hat_and_lambda_probability}
    \end{align}

\item \boldsymbol{$\hat{r}^{(n)}_1$}: Note that
\begin{align*}
\hat{r}^{(n)}_1 &= \left( \frac{1}{U_n}+1 \right)\frac{n\hat{\lambda}^{(n)}_1}{I_ne^{\hat{\lambda}^{(n)}_1\gamma_n}}
\\
& =  \left( \frac{1}{U_n}+1 \right)\hat{\lambda}^{(n)}_1\frac{n^{1-\alpha}}{I_n}e^{\hat{\lambda}^{(n)}_1(\zeta_n-\gamma_n)}n^{\alpha}e^{-\hat{\lambda}^{(n)}_1\zeta_n}\\
& =  \left( \frac{1}{U_n}+1 \right)\hat{\lambda}^{(n)}_1\frac{n^{1-\alpha}}{I_n}e^{\hat{\lambda}^{(n)}_1(\zeta_n-\gamma_n)}\left(n^{\alpha}e^{-\lambda_1\zeta_n}\right)e^{(\lambda_1-\hat{\lambda}^{(n)}_1)\zeta_n}.
\end{align*}
By \eqref{eqn:u_n_and_u_in_probability}, \eqref{eqn:lambda_1_hat_and_lambda_probability}, Theorem \ref{cip_In}, Theorem \ref{cip_gamma_faster}, and \eqref{Order: zeta_n}, we can obtain that each term in the above expression $(\left( \frac{1}{U_n}+1 \right), \hat{\lambda}^{(n)}_1, \frac{n^{1-\alpha}}{I_n}, e^{\hat{\lambda}^{(n)}_1(\zeta_n-\gamma_n)}, \left(n^{\alpha}e^{-\lambda_1\zeta_n}\right), e^{(\lambda_1-\hat{\lambda}^{(n)}_1)\zeta_n})$ converges in probability to the corresponding limit  $(\frac{\lambda_0-\lambda_1}{\lambda_0}, \lambda_1, -\frac{\lambda_0 r_1}{\lambda_1}, 1, \frac{1}{\lambda_1-\lambda_0}, 1)$, respectively. Hence, $\hat{r}^{(n)}_1 $ converges in probability to $r_1$ by the Continuous Mapping Theorem.

    \item \boldsymbol{$\hat{\alpha}^{(n)}_1$}: By Theorem \ref{cip_In}, we have 
    \begin{align*}
         \lim_{n\to\infty}\PP\left(\left| \frac{1}{n^{1-\alpha}}I_n\left(\gamma_n\right)+\frac{\lambda_1}{\lambda_0 r_1} \right| <\epsilon \right) = 1
    \end{align*}
    which implies that 
    $$
   \lim_{n\to\infty} P\left(c n^{1-\alpha}\leq  I_n(\gamma_n) \leq C n^{1-\alpha}\right) =1,
    $$
    where $c$ and $C$ are some positive constants. It follows that
    $$
    \hat{\alpha}^{(n)} - \alpha = 1 - \log_nI_n -\alpha + \log_n\left( \frac{\hat{\lambda}^{(n)}_1}{- \hat{\lambda}^{(n)}_0\hat{r}^{(n)}_1} \right) \xrightarrow{p} 0.
    $$
    \end{itemize}
\end{proof}

\newpage
\bibliographystyle{plain}
\bibliography{ref}

\newpage

\vspace{0.5in}



\end{document}